\documentclass{article}

\usepackage{arxiv}

\usepackage[utf8]{inputenc}
\usepackage[T1]{fontenc}
\usepackage{amsmath}
\usepackage{amsthm}
\usepackage{newtxtext}
\usepackage{newtxmath}
\usepackage{bm}
\usepackage{graphicx}
\usepackage{booktabs}
\usepackage{microtype}
\usepackage{natbib}
\usepackage{url}
\usepackage{hyperref}
\usepackage{doi}
\usepackage{cleveref}

\allowdisplaybreaks[4]
\bibpunct[, ]{(}{)}{,}{a}{}{,}
\theoremstyle{plain}
\newtheorem{theorem}{Theorem}
\newtheorem{proposition}{Proposition}
\newtheorem{lemma}{Lemma}
\newtheorem{corollary}{Corollary}
\theoremstyle{definition}
\newtheorem{assumption}{Assumption}
\newtheorem{definition}{Definition}
\theoremstyle{remark}
\newtheorem{remark}{Remark}

\title{Regret-Based Max-Pressure Traffic Signal Control under Uncertain Traffic States}

\author{%
Chaopeng Tan* \qquad Qiongdan Hu \qquad Meng Wang\\
\normalfont Chair of Traffic Process Automation, Technische Universit\"at Dresden\\
\normalfont Hettnerstra\ss e 3, 01069 Dresden, Germany\\
\normalfont \texttt{chaopeng.tan@tu-dresden.de, qiongdan.hu@tu-dresden.de, meng.wang@tu-dresden.de}
}
\date{}

\renewcommand{\shorttitle}{Regret-Based Max-Pressure Control}

\hypersetup{
  pdftitle={Regret-Based Max-Pressure Traffic Signal Control under Uncertain Traffic States},
  pdfauthor={Chaopeng Tan, Qiongdan Hu, Meng Wang},
  pdfkeywords={max-pressure signal control, traffic-state uncertainty, phase regret, distributionally robust optimization, queue stability}
}

\begin{document}
\maketitle

\begin{abstract}
Max-pressure (MP) traffic signal control selects phases according to pressure scores constructed from estimated traffic states. In practice, these states are uncertain, and the consequence of an error depends on whether it changes the phase ranking and how much pressure is lost. To characterize these decision consequences, this paper develops a decision-oriented \emph{Regret-Based Max-Pressure} framework that represents traffic-state uncertainty through phase-selection regret, defined as the pressure loss of a candidate phase relative to the pressure-maximizing phase under each plausible latent state. We propose two control methods: Stochastic Regret MP, which minimizes expected phase-selection regret under a calibrated residual distribution, and Distributionally Robust (DR) Regret MP, which minimizes worst-case expected regret over a Wasserstein ambiguity set. A key modeling contribution is to robustify phase-selection regret rather than absolute pressure: because the regret representation is invariant to state-dependent pressure components shared by all candidate phases, the resulting robustification concentrates on relative phase quality and avoids sensitivity to scenario variation that does not affect phase ranking. We further translate uncertainty-representation error into an explicit queue-linear pressure-loss bound and establish strong queue stability when the available net-service reserve exceeds the resulting loss coefficient. Numerical experiments show that Stochastic Regret MP and DR Regret MP reduce normalized true-state pressure regret by about 70\% relative to Point-estimate MP. Under the same uncertainty information, DR Regret MP also produces lower network queues than a controller that directly robustifies absolute phase pressure. In microscopic SUMO experiments on a real-world network with 28 intersections, Stochastic Regret MP and DR Regret MP reduce mean delay by 18.9\% and 14.9\%, respectively, relative to Point-estimate MP, and both retain their benefits under detector undercount and detector failure. These findings show that representing uncertainty through its decision consequences can improve the reliability of traffic control under imperfect state information, providing a general design principle for data-driven transportation decision-making under uncertainty.
\end{abstract}

\section{Introduction}
\label{sec:introduction}

Max-pressure (MP) traffic signal control selects phases by comparing pressure scores constructed from local traffic states. Originating from maximum-throughput scheduling \citep{tassiulas1992stability}, MP was adapted to signalized road networks as a decentralized control principle with queue-stability guarantees for demand within the admissible region \citep{wongpiromsarn2012distributed,varaiya2013max}. Subsequent studies further extended the decentralized MP framework and its network-level analysis \citep{le2015decentralized,levin2023summary,liu2024nmp}. The combination of local computation and network-level stability makes MP attractive for large-scale traffic control. At the same time, because MP determines phase actions by ranking state-dependent pressure scores, the resulting decisions depend directly on the traffic-state information used to construct those scores.

Modern MP formulations increasingly rely on richer, and often estimated, traffic information. Beyond conventional queue lengths, existing formulations incorporate vehicle positions \citep{li2019position}, travel-time or delay information \citep{liu2022delay,zoabi2025ttd}, and connected-vehicle or heterogeneous observations for traffic-state estimation under partial observation \citep{cao2020queue,tan2026cvmp,tan2026transitmp,tan2026privatemp}. As a result, the pressures used for phase selection may depend on estimated queues, turning proportions, and effective movement service. Existing MP studies have considered unavailable turning information and noisy or estimated queues \citep{gregoire2014unknown,charalambous2023noisy,liaquat2024noisy}, but these treatments largely address individual information components. They do not provide a common representation of how heterogeneous traffic-state errors jointly affect the relative ranking of candidate phases and the pressure loss associated with the resulting phase decision.

The central difficulty is that the magnitude of a traffic-state error does not by itself determine its consequence for MP control. A sizeable perturbation can alter several phase-pressure values without changing their ordering, whereas a smaller perturbation can reverse the preferred phase when competing pressures are close. The relevant consequence of uncertainty is therefore decision-relative: it depends on how the uncertain state changes the relative quality of candidate phases and how much pressure is lost when a nonoptimal phase is selected. More broadly, traffic signal control has addressed uncertainty through stochastic, reliability-based, and robust optimization, including uncertain arrivals and demand \citep{jia2025adaptive,pedroso2024feedback}, traffic-flow variation \citep{yin2008robust,tan2025robust}, uncertain turning and network interactions \citep{fei2023stochastic,wan2024coordination}, and sensing or estimation errors \citep{jia2023uncertainty,tan2025contextual}. These approaches incorporate uncertainty into signal timing, prediction, network flows, or estimated states, but they do not directly characterize uncertainty in pressure-based phase ranking. Regret-based optimization provides a natural decision-relative perspective by measuring the loss of a chosen action relative to the action that would be optimal under the realized uncertain state \citep{kouvelis1997robust,aissi2009minmax}. Distributionally robust optimization (DRO) further accounts for ambiguity in the probability distribution of uncertain states \citep{delage2010distributionally,gao2023distributionally}, with Wasserstein ambiguity sets providing a data-driven construction \citep{esfahani2018data}; recent studies have also combined distributional robustness with regret-based criteria \citep{altaha2023distributionally,cho2024wasserstein,kargin2024infinite}. For MP control, however, two questions remain unresolved: how to define a decision loss that is consistent with pressure-based phase ranking, and how robustifying this decision-relative loss differs from directly robustifying absolute phase pressure under distributional ambiguity.

To address these questions, we develop a decision-oriented \emph{Regret-Based Max-Pressure} framework. Historical estimation residuals are used to reconstruct plausible latent traffic states, and \emph{phase-selection regret} measures the pressure loss of a candidate phase relative to the pressure-maximizing phase under each latent state. Based on this decision-loss representation, Stochastic Regret MP minimizes expected phase-selection regret under a calibrated residual distribution, while Distributionally Robust (DR) Regret MP minimizes worst-case expected regret over a Wasserstein ambiguity set.
The paper makes four main contributions.
\begin{itemize}

    \item \textbf{A decision-relative representation of traffic-state uncertainty for MP control.}
    We formulate phase-selection regret to characterize traffic-state uncertainty through its consequence for the phase decision rather than through state-error magnitude alone. This representation places heterogeneous uncertainty on a common pressure-loss scale and provides the basis for Stochastic Regret MP and Wasserstein DR Regret MP.

    \item \textbf{A decision-focused robustification of pressure-based phase selection.}
    A key modeling contribution is to robustify phase-selection regret rather than absolute phase pressure. Regret preserves the expected-score decision under a fixed distribution and, under distributional ambiguity, filters out state-dependent pressure variation shared by all candidate phases. The resulting robustification therefore concentrates on relative phase quality rather than on scenario variation that does not affect phase ranking, as may occur with direct absolute-pressure robustification.

    \item \textbf{A stability characterization under uncertain traffic states.}
    We translate uncertainty-representation error into an explicit pressure-loss bound with a queue-linear component and establish strong queue stability when the available net-service reserve exceeds the resulting loss coefficient. The same analysis places Regret MP and Point-estimate MP on a common pressure-loss scale and provides sufficient conditions for comparing their provable stability regions.

    \item \textbf{Mechanistic and network-level evidence of the benefits of regret-based control.}
    Numerical experiments show that Stochastic Regret MP and DR Regret MP substantially reduce true-state phase-selection loss and network queues relative to Point-estimate MP. They also show that, under the same uncertainty information, regret-based robustification achieves better network performance than direct robustification of absolute phase pressure. Microscopic experiments further demonstrate reductions in delay and congestion on a 28-intersection network, with the benefits persisting under detector undercount and detector failure.

\end{itemize}

The remainder of this paper is organized as follows. Section~\ref{sec:mp_uncertainty} formulates MP control with uncertain traffic states. Section~\ref{sec:regret_control} develops the regret-based controllers and their structural properties. Section~\ref{sec:theory} establishes the network-stability results. Section~\ref{sec:numerical} presents numerical experiments, and Section~\ref{sec:microscopic} evaluates the framework in microscopic simulation. Section~\ref{sec:conclusion} concludes the paper.

\section{Max-Pressure Control under Uncertain Traffic States}
\label{sec:mp_uncertainty}
\label{sec:problem}

The proposed framework retains the physical queue dynamics and decentralized phase-selection structure of conventional MP. Its point of departure is informational: the phase decision is made from an estimated traffic state, whereas the actual pressure is determined by a latent state that may differ from that estimate. This section defines the network dynamics and standard MP phase score, then distinguishes the available estimate from the latent traffic state. The construction of uncertainty scenarios and their conversion into phase-selection regret are developed in Section~\ref{sec:regret_control}.

Throughout the paper, calligraphic letters denote sets, bold letters denote vectors or matrices, and hats denote online estimates. The scenario index is $k$, while superscripts $\mathrm{T}$, $\mathrm{P}$, $\mathrm{Sto}$, $\mathrm{DR}$, and $\mathrm{Set}$ identify the controller or uncertainty treatment. Blackboard-bold $\mathbb P$ denotes a probability law, whereas the roman superscript $\mathrm P$ denotes Point-estimate MP. A superscript $\star$ is reserved for an optimizer.

\subsection{Network and traffic dynamics}
\label{sec:network_dynamics}

Following the conventional store-and-forward representation of signalized traffic networks used in MP control \citep{varaiya2013max}, we model traffic at the movement level through queues, routing proportions, and saturation service.
Consider a directed road network with a set of signalized intersections $\mathcal N$ and a set of permissible movements $\mathcal M$. For intersection $n\in\mathcal N$, let $\mathcal M_n\subseteq\mathcal M$ denote the movements controlled by that intersection. A movement represents a queue of vehicles sharing the same immediate downstream destination. Let $Q_m(t)\geq0$ be the number of vehicles in movement queue $m$ at the beginning of control interval $t$, and let $\bm Q(t)$ collect all movement queues. Intersection $n$ has a finite phase set $\mathcal S_n$. A phase vector $\bm s_n\in\mathcal S_n$ has component $s_m\in\{0,1\}$ for $m\in\mathcal M_n$, where $s_m=1$ activates movement $m$. The network phase vector is $\bm s(t)=\operatorname{col}(\bm s_n(t):n\in\mathcal N)$ and belongs to the Cartesian product $\mathcal S=\prod_{n\in\mathcal N}\mathcal S_n$.

Let $A_m(t)$ denote exogenous arrivals to movement $m$, and let $\bm A(t)$ collect all exogenous arrivals. The routing matrix $\bm R(t)=[r_{ab}(t)]$ is defined such that $r_{ab}(t)$ denotes the fraction of vehicles discharged from movement $a$ that subsequently enter downstream movement $b$. Thus, $r_{ab}(t)\geq0$ and $\sum_b r_{ab}(t)\leq1$; a strict inequality allows a fraction of the discharged traffic to leave the modeled network.
Let $c_m(t)>0$ denote the number of vehicles that can potentially depart from movement $m$ during control interval $t$ when the movement is actuated, and let $\bm C(t)=\operatorname{diag}\{c_m(t):m\in\mathcal M\}$. The quantity $c_m(t)$ represents the prevailing saturation service available to movement $m$ during the interval and may vary with local discharge conditions. The realized discharge is
\(D_m(t) = \min\{Q_m(t),c_m(t)s_m(t)\},\)
and $\bm D(t)$ collects all realized discharges. Let $\bm I$ denote the identity matrix of compatible dimension. With these definitions, conservation of vehicles gives the network queue dynamics
\begin{equation}
    \bm Q(t+1)
    =
    \bm Q(t)
    +
    \bm A(t)
    -
    [\bm I-\bm R^{\top}(t)]\bm D(t).
    \label{eq:queue_dynamics}
\end{equation}
Equation~\eqref{eq:queue_dynamics} is the physical backbone of the subsequent stability analysis: each realized discharge is removed from its current movement, while the routed fraction is added to downstream movements. The control problem is therefore to allocate the available saturation service across feasible phases so as to create sufficient net service relative to the external arrivals.

\subsection{Max-pressure phase selection}
\label{sec:mp_phase_selection}

MP converts the network state in~\eqref{eq:queue_dynamics} into a scalar score for each feasible phase. The score rewards service to a large upstream queue and discounts service that sends vehicles toward heavily occupied downstream movements. Under the latent traffic state, the pressure released by network phase $\bm s$ is
\begin{equation}
    H(\bm s,t)
    =
    \bm Q^{\top}(t)
    [\bm I-\bm R^{\top}(t)]
    \bm C(t)\bm s.
    \label{eq:true_pressure}
\end{equation}
Equivalently, each served movement $a$ contributes $c_a(t)s_a[Q_a(t)-\sum_b r_{ab}(t)Q_b(t)]$, i.e., its available saturation service multiplied by the upstream-minus-routed-downstream queue imbalance. Because each movement is controlled by one intersection, the network pressure decomposes as $H(\bm s,t)=\sum_{n\in\mathcal N}H_n(\bm s_n,t)$. This separability is the computational reason that MP can make local phase decisions while acting on a network-level Lyapunov drift. The true-state MP reference controller is therefore
\begin{equation}
    \bm s_n^{\mathrm T}(t)
    \in
    \arg\max_{\bm s_n\in\mathcal S_n}
    H_n(\bm s_n,t),
    \qquad n\in\mathcal N.
    \label{eq:true_state_mp}
\end{equation}
The term ``true-state MP'' is used throughout the paper to denote the MP phase rule evaluated at the latent traffic state.
The formulation above uses unweighted movement queues for clarity. The same pressure and regret construction extends to any fixed positive diagonal queue weighting, including the link-length normalization used in the microscopic experiments. Appendix~\ref{app:weighted_stability} gives the complete weighted formulation and stability result.

The true-state rule provides the reference decision when the latent network state is known. An implementable controller, however, acts on measurements or estimates of queues, turning proportions, and effective service. The next subsection therefore separates the available estimate from the latent state that generates the actual pressure.

\subsection{Estimated traffic state and point-estimate MP}
\label{sec:uncertain_state}

The controller does not generally observe the latent traffic state $\bm x(t)=(\bm Q(t),\bm R(t),\bm C(t))$ directly. Instead, it receives the estimate
$\widehat{\bm x}(t)=(\widehat{\bm Q}(t),\widehat{\bm R}(t),\widehat{\bm C}(t))$.
We represent the discrepancy between the estimate and the latent state by the joint residual vector
$\bm\xi(t)=\operatorname{col}(\bm\xi^Q(t),\bm\xi^R(t),\bm\xi^C(t))\in\Xi$,
where $\Xi$ denotes the residual support, so that dependence among queue, routing, and service errors can be retained.

Queue and saturation-service errors are represented multiplicatively,
$\widehat Q_m(t)=[1+\xi_m^Q(t)]Q_m(t)$ and 
$\widehat c_m(t)=[1+\xi_m^C(t)]c_m(t),$
whereas routing uncertainty is represented additively as
$\widehat{\bm R}(t)=\bm R(t)+\bm\Xi^R(t),$
with $\bm\Xi^R(t)$ denoting the matrix form of $\bm\xi^R(t)$. The multiplicative form reflects that queue length and service capacity are nonnegative scale quantities, while routing entries are bounded proportions. We assume
$1+\xi_m^Q\geq\underline\chi_Q>0$ and
$1+\xi_m^C\geq\underline\chi_C>0$
for all $\bm\xi\in\Xi$, ensuring that the latent queue and service states can be uniquely reconstructed from the estimate and a residual realization. Routing matrices are restricted to the feasible set defined by nonnegative, row-substochastic turning proportions and admissible network movements.

Define the pressure evaluated at the current estimate as
$$
\widehat H(\bm s,t)
=
\widehat{\bm Q}^{\top}(t)
[\bm I-\widehat{\bm R}^{\top}(t)]
\widehat{\bm C}(t)\bm s.
$$
A controller that directly treats the current estimate as the operative traffic state applies
\begin{equation}
    \bm s^{\mathrm P}(t)
    \in
    \arg\max_{\bm s\in\mathcal S}
    \widehat H(\bm s,t).
    \label{eq:point_estimate_mp}
\end{equation}
The point-estimate rule in~\eqref{eq:point_estimate_mp} is referred to as Point-estimate MP throughout the paper, and its component at intersection $n$ is denoted by $\bm s_n^{\mathrm P}(t)$. It applies the MP phase rule directly to the available traffic estimate without aggregating residual uncertainty; ``Point'' is used as a compact label in figures and tables where space is limited.

Historical paired estimates and corresponding latent traffic states are used offline to construct a joint empirical residual distribution.
For online control, the empirical distribution is compressed to a finite set of representative joint residual scenarios and projected to the components required by each intersection. Let $\mathcal K_n=\{1,\ldots,K_n\}$ denote the index set of the retained local scenarios. The resulting local nominal distribution is
\begin{equation}
    \widehat{\mathbb P}_n
    =
    \sum_{k\in\mathcal K_n}
    \widehat\pi_{n,k}
    \delta_{\bar{\bm\xi}_{n}^{k}},
    \label{eq:local_nominal_distribution}
\end{equation}
where $\bar{\bm\xi}_{n}^{k}$ is the $k$-th retained local residual scenario, $\widehat\pi_{n,k}$ is its calibrated probability, and $\delta_{\bar{\bm\xi}_{n}^{k}}$ denotes the Dirac measure, i.e., a unit point mass concentrated at $\bar{\bm\xi}_{n}^{k}$. Thus, $\widehat{\mathbb P}_n$ is a discrete distribution over joint error patterns: queue, routing, and service errors within one scenario are kept together rather than treated as independent marginal perturbations. We write $\widehat{\bm\pi}_n=(\widehat\pi_{n,1},\ldots,\widehat\pi_{n,K_n})$ for the corresponding probability vector. The detailed residual construction, scaling, scenario reduction, and local projection are given in Appendix~\ref{app:residual_calibration}. Section~\ref{sec:regret_control} uses these scenarios to reconstruct plausible latent states and evaluate the consequence of each candidate phase.

\section{Regret-Based Max-Pressure Control}
\label{sec:regret_control}

The proposed framework represents uncertain traffic states through their consequences for pressure-based phase decisions. This section defines phase-selection regret, develops Stochastic and Distributionally Robust (DR) Regret MP, and establishes the structural relationships among regret-based and absolute-pressure robustification.

\subsection{Decision-relative representation of uncertain traffic states}
\label{sec:phase_regret_definition}
\label{sec:general_regret}

This subsection converts traffic-state uncertainty into a decision loss for phase selection. Each retained residual scenario reconstructs a plausible latent traffic state, which yields scenario-specific phase pressures; the pressure gap from the best phase defines phase-selection regret. The resulting regret matrix is the common input to the two controllers developed in Section~\ref{sec:stochastic_regret}.

Consider a finite action set $\mathcal A$ with state-dependent score $S(a,\bm\xi)$. Let $M(\bm\xi)=\max_{u\in\mathcal A}S(u,\bm\xi)$ and define the decision regret as $L(a,\bm\xi)=M(\bm\xi)-S(a,\bm\xi)$. The first property needed below is that this representation preserves the nominal score-based decision under a fixed distribution.

\begin{proposition}[Expected-regret equivalence for score-based decisions]
\label{prop:fixed_distribution_equivalence}
Let $\mathbb P$ be any fixed probability distribution under which $M(\bm\xi)$ and $S(a,\bm\xi)$ are integrable for all $a\in\mathcal A$. Then
\begin{equation}
    \arg\min_{a\in\mathcal A}
    \mathbb E_{\mathbb P}[L(a,\bm\xi)]
    =
    \arg\max_{a\in\mathcal A}
    \mathbb E_{\mathbb P}[S(a,\bm\xi)].
    \label{eq:fixed_distribution_equivalence}
\end{equation}
\end{proposition}

By the definition $L(a,\bm\xi)=M(\bm\xi)-S(a,\bm\xi)$, the term $\mathbb E_{\mathbb P}[M(\bm\xi)]$ is independent of the candidate action, so minimizing expected regret and maximizing expected score have the same optimizer set. For MP, this means that using phase-selection regret under a fixed distribution preserves the probability-weighted pressure decision. This equivalence provides the basis for representing traffic-state uncertainty through regret before the scenario-specific losses are aggregated in the controllers below. The proof is given in Appendix~\ref{app:regret_structural_proofs}.

We now instantiate this representation for MP control. For intersection $n$, let $\mathcal L_n$ contain the controlled and immediate downstream movements required by its pressure calculation, let $\Xi_n$ denote the corresponding local residual support, and let $\mathcal K_n=\{1,\ldots,K_n\}$ index the retained local residual scenarios. Let $\mathcal R_n$ denote the closed convex set of feasible local routing matrices induced by the routing constraints in Section~\ref{sec:uncertain_state}, and let $\Pi_{\mathcal R_n}$ denote Euclidean projection onto this set. The estimated local state is
\(\widehat{\bm x}_n(t)=(\widehat{\bm Q}_{\mathcal L_n}(t),\widehat{\bm R}_n(t),\widehat{\bm C}_n(t)).\)
For scenario $k$, inversion of the residual relations in Section~\ref{sec:uncertain_state} gives
\begin{equation}
    Q_m^k(t)
    =
    \frac{\widehat Q_m(t)}{1+\bar\xi_{m,k}^Q},
    \qquad
    c_m^k(t)
    =
    \frac{\widehat c_m(t)}{1+\bar\xi_{m,k}^C},
    \label{eq:scenario_qc_reconstruction}
\end{equation}
and
\begin{equation}
    \bm R_n^k(t)
    =
    \Pi_{\mathcal R_n}
    \left[
        \widehat{\bm R}_n(t)
        -
        \overline{\bm\Xi}_{n,k}^R
    \right],
    \label{eq:scenario_r_reconstruction}
\end{equation}
where $\overline{\bm\Xi}_{n,k}^R$ is the matrix form of the routing residual in scenario $k$ and $\Pi_{\mathcal R_n}$ projects onto the feasible routing set. Equations~\eqref{eq:scenario_qc_reconstruction}--\eqref{eq:scenario_r_reconstruction} map each retained residual scenario to a plausible latent traffic state. These quantities define the reconstructed latent state
\(\bm x_n^k(t)=\mathcal T_n(\widehat{\bm x}_n(t),\bar{\bm\xi}_n^k),\)
where $\mathcal T_n$ denotes the local latent-state reconstruction map.

Each reconstructed state is then used to evaluate the feasible phases. For candidate phase $\bm s_n\in\mathcal S_n$, its pressure under scenario $k$ is
$$H_{n,k}(\bm s_n,t):=\sum_{a\in\mathcal M_n}c_a^k(t)s_a\!\left[Q_a^k(t)-\sum_{b\in\mathcal L_n}r_{ab}^k(t)Q_b^k(t)\right].$$
Let
$H_{n,k}^{\max}(t)=\max_{\bm u_n\in\mathcal S_n}H_{n,k}(\bm u_n,t).$
The phase-selection regret is
\begin{equation}
    L_{n,k}(\bm s_n,t)
    =
    H_{n,k}^{\max}(t)
    -
    H_{n,k}(\bm s_n,t)
    \geq0.
    \label{eq:scenario_regret}
\end{equation}
Thus, each scenario maps heterogeneous traffic-state uncertainty to a phase-wise pressure loss relative to the scenario-best phase. Figure~\ref{fig:phase_regret_schematic} illustrates this construction.

\begin{figure}
\centering
\includegraphics[width=0.9\linewidth]{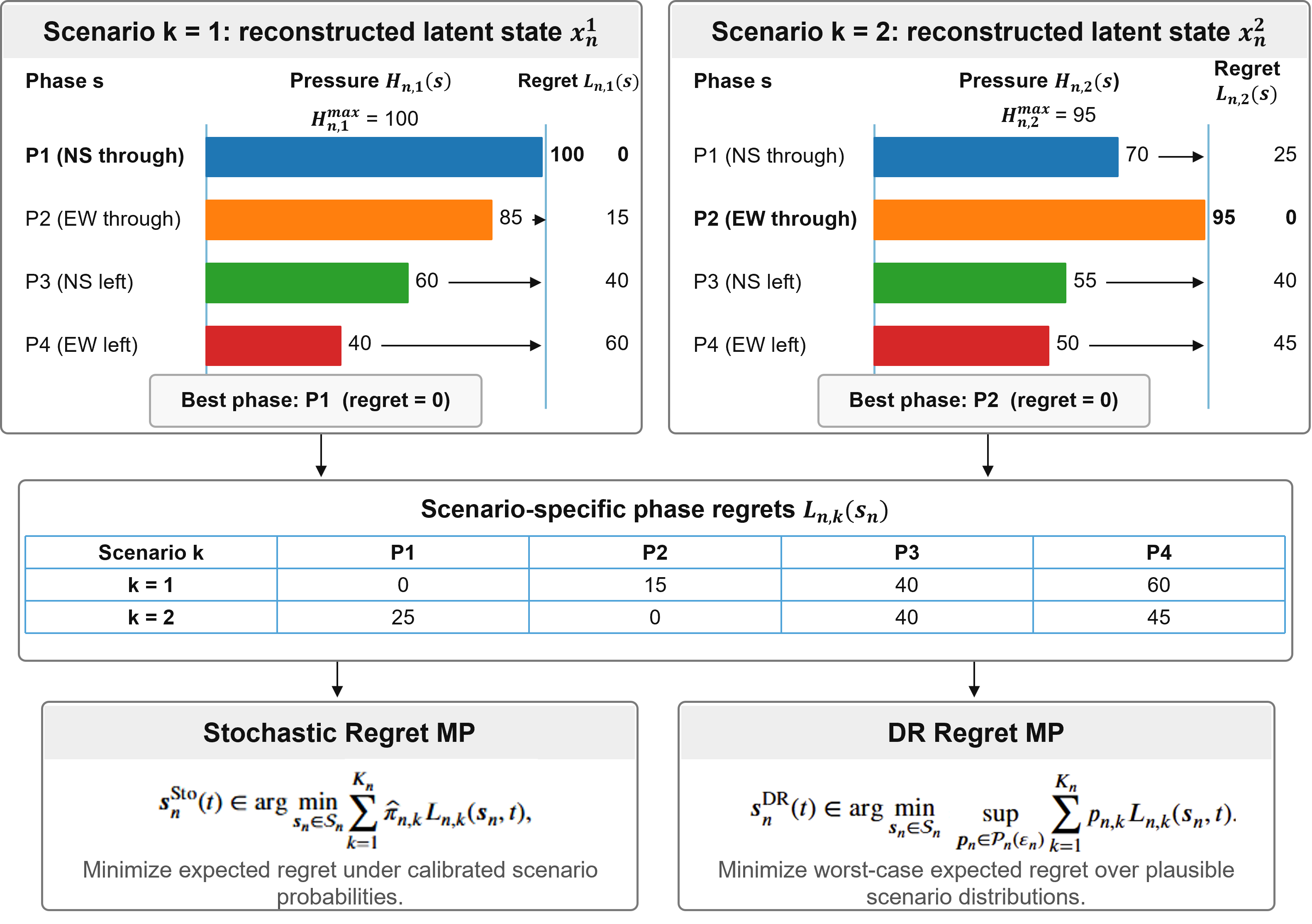}
\caption{Illustration of phase-selection regret. Each retained residual scenario reconstructs a plausible latent state and its phase pressures. The pressure gap from the scenario-best phase defines $L_{n,k}(\bm s_n)$; Stochastic Regret MP aggregates these regrets using calibrated probabilities, whereas DR Regret MP uses a Wasserstein ambiguity set. Values are illustrative.}
\label{fig:phase_regret_schematic}
\end{figure}

\subsection{Stochastic and distributionally robust regret control}
\label{sec:stochastic_regret}
\label{sec:dr_regret}

Given the regret matrix constructed above, this subsection defines two phase-selection rules. 
Let $\mathcal S_n=\{\bm s_{n,1},\ldots,\bm s_{n,N_n^{\mathrm{ph}}}\}$, where $N_n^{\mathrm{ph}}=|\mathcal S_n|$. Define the phase--scenario pressure matrix $\bm H_n(t)=[H_{n,k}(\bm s_{n,j},t)]_{j,k}$ and regret matrix $\bm L_n(t)=[L_{n,k}(\bm s_{n,j},t)]_{j,k}$. The two proposed controllers use the same regret matrix and differ only in how each phase row is aggregated across scenarios. For $\psi\in\{\mathrm{Sto},\mathrm{DR}\}$, let $G_n^\psi(\bm s_n,t)$ denote the aggregate regret score and write
$$\bm s_n^\psi(t)\in\arg\min_{\bm s_n\in\mathcal S_n}G_n^\psi(\bm s_n,t).$$
Stochastic Regret MP aggregates scenario regrets using the calibrated probabilities, whereas DR Regret MP minimizes worst-case expected regret over nearby scenario distributions.

\paragraph{Stochastic Regret MP.}
Using the calibrated scenario probabilities directly gives the expected-regret score
\begin{equation}
    G_n^{\mathrm{Sto}}(\bm s_n,t)
    =
    \sum_{k=1}^{K_n}
    \widehat\pi_{n,k}
    L_{n,k}(\bm s_n,t),
    \label{eq:stochastic_regret}
\end{equation}
and Stochastic Regret MP selects
\begin{equation}
    \bm s_n^{\mathrm{Sto}}(t)
    \in
    \arg\min_{\bm s_n\in\mathcal S_n}
    G_n^{\mathrm{Sto}}(\bm s_n,t).
    \label{eq:stochastic_controller}
\end{equation}
By Proposition~\ref{prop:fixed_distribution_equivalence},
\begin{equation}
    \arg\min_{\bm s_n\in\mathcal S_n}
    G_n^{\mathrm{Sto}}(\bm s_n,t)
    =
    \arg\max_{\bm s_n\in\mathcal S_n}
    \sum_{k=1}^{K_n}
    \widehat\pi_{n,k}H_{n,k}(\bm s_n,t),
    \label{eq:sto_expected_equivalence}
\end{equation}
so Stochastic Regret MP preserves the probability-weighted MP decision under the calibrated distribution.

\paragraph{Distributionally Robust Regret MP.}
DR Regret MP allows the scenario probabilities to vary within a Wasserstein ambiguity set $\mathcal P_n(\varepsilon_n)$ centered at the calibrated distribution. Its worst-case expected regret score is
\begin{equation}
    G_n^{\mathrm{DR}}(\bm s_n,t;\varepsilon_n)
    =
    \sup_{\bm p_n\in\mathcal P_n(\varepsilon_n)}
    \sum_{k=1}^{K_n}
    p_{n,k}L_{n,k}(\bm s_n,t),
    \label{eq:dr_regret_score}
\end{equation}
and DR Regret MP selects
\begin{equation}
    \bm s_n^{\mathrm{DR}}(t)
    \in
    \arg\min_{\bm s_n\in\mathcal S_n}
    G_n^{\mathrm{DR}}(\bm s_n,t;\varepsilon_n).
    \label{eq:dr_controller}
\end{equation}
To make the distributional protection in~\eqref{eq:dr_regret_score} operational, we need to specify which alternative scenario distributions are sufficiently close to the calibrated distribution. Because the retained scenarios occupy different locations in the residual space, this proximity should reflect both the amount of probability mass redistributed and the distance over which it is moved. We therefore use a discrete Wasserstein distance to characterize probability redistribution over the retained residual scenarios.

Let $d_n$ denote the local residual-space ground metric induced by the scaled residual metric in Appendix~\ref{app:residual_calibration}, and set $d_{n,ij}=d_n(\bar{\bm\xi}_n^i,\bar{\bm\xi}_n^j)$. For an alternative probability vector $\bm p_n\in\Delta_{K_n}$, where $\Delta_{K_n}$ is the probability simplex, define the discrete Wasserstein distance
\begin{equation}
    W_n(\bm p_n,\widehat{\bm\pi}_n)
    =
    \min_{\bm\varphi\succeq\bm0}
    \sum_{i=1}^{K_n}\sum_{j=1}^{K_n}
    d_{n,ij}\varphi_{ij}
    \label{eq:discrete_wasserstein}
\end{equation}
    $$\text{s.t.}\quad
    \sum_j\varphi_{ij}
    =\widehat\pi_{n,i},
    \quad i=1,\ldots,K_n,$$
    $$\sum_i\varphi_{ij}
    =p_{n,j},
    \quad j=1,\ldots,K_n,$$

where $\varphi_{ij}$ is the probability mass transported from nominal scenario $i$ to scenario $j$, and $d_{n,ij}$ is the corresponding transport cost. The ambiguity set is
\begin{equation}
    \mathcal P_n(\varepsilon_n)
    =
    \left\{
        \bm p_n\in\Delta_{K_n}:
        W_n(\bm p_n,\widehat{\bm\pi}_n)
        \leq\varepsilon_n
    \right\},
    \label{eq:local_ambiguity_set}
\end{equation}
where $\varepsilon_n\geq0$ is the transport budget. At $\varepsilon_n=0$, the calibrated distribution is recovered; larger values allow greater redistribution of probability mass. The ground metric and ambiguity radius are calibrated offline from residual data as described in Appendix~\ref{app:residual_calibration}.

At each decision epoch, DR Regret MP must evaluate the worst-case expected regret over this ambiguity set for every candidate phase. Because the retained support is finite, the Wasserstein constraint can be expressed directly in the transport variables $\bm\varphi$. This converts the inner maximization in~\eqref{eq:dr_regret_score} into the following finite-dimensional linear program:
\begin{equation}
    G_n^{\mathrm{DR}}(\bm s_n,t;\varepsilon_n)
    =
    \max_{\bm\varphi\succeq\bm0}
    \sum_{i=1}^{K_n}\sum_{j=1}^{K_n}
    \varphi_{ij}L_{n,j}(\bm s_n,t)
    \label{eq:dr_regret_transport}
\end{equation}
    $$\text{s.t.}\quad
    \sum_j\varphi_{ij}
    =\widehat\pi_{n,i},
    \quad i=1,\ldots,K_n,$$
    $$\sum_{i,j}d_{n,ij}\varphi_{ij}
    \leq\varepsilon_n.
    \nonumber$$

This form makes the worst-case redistribution explicit: transported mass contributes the regret of its destination scenario and consumes the corresponding transport budget. For online evaluation, the same problem admits an exact one-dimensional reformulation.

\begin{lemma}[Exact scalar reformulation of DR regret]
\label{lem:dr_regret_dual}
For every phase $\bm s_n$, define
$g_{n,i}(\nu;\bm s_n,t)=\max_{j=1,\ldots,K_n}\{L_{n,j}(\bm s_n,t)-\nu d_{n,ij}\}$.
Then
\begin{equation}
    G_n^{\mathrm{DR}}(\bm s_n,t;\varepsilon_n)
    =
    \min_{\nu\geq0}
    \{
        \nu\varepsilon_n
        +
        \sum_{i=1}^{K_n}
        \widehat\pi_{n,i}g_{n,i}(\nu;\bm s_n,t)
    \}.
    \label{eq:dr_regret_dual}
\end{equation}
\end{lemma}

The transport problem in~\eqref{eq:dr_regret_transport} is a finite linear program. Dualizing its transport-budget constraint introduces the scalar $\nu$, while the row-wise dual constraints reduce to the maxima collected in $g_{n,i}$. Lemma~\ref{lem:dr_regret_dual} therefore leaves the DR regret score unchanged but converts its evaluation into a one-dimensional convex piecewise-linear minimization for each candidate phase. This exact reformulation is the computational form used by DR Regret MP. The proof is given in Appendix~\ref{app:regret_structural_proofs}.

\subsection{Structural Advantage of Decision-Relative Robustification}
\label{sec:decision_relative_robustification}

This subsection characterizes the structural difference between robustifying decision regret and robustifying the absolute score. Let $\mathfrak P$ denote a generic ambiguity set of probability laws and define the robust-regret criterion
$$
J_{\mathrm{reg}}(a;\mathfrak P)
=\sup_{\mathbb P\in\mathfrak P}\mathbb E_{\mathbb P}[M(\bm\xi)-S(a,\bm\xi)]
$$
and the robust absolute-score criterion
$$
J_{\mathrm{abs}}(a;\mathfrak P)
=\inf_{\mathbb P\in\mathfrak P}\mathbb E_{\mathbb P}[S(a,\bm\xi)].
$$

The first structural property concerns state-dependent score components shared by all actions.

\begin{proposition}[Scenario-wise common-shift invariance]
\label{prop:common_shift_invariance}
Let $\widetilde S(a,\bm\xi)=S(a,\bm\xi)+c(\bm\xi)$, where $c(\bm\xi)$ is independent of $a$. Then the corresponding regret satisfies $\widetilde L(a,\bm\xi)=L(a,\bm\xi)$ for every $a$ and $\bm\xi$. Consequently, fixed-distribution expected regret, worst-case expected regret over any fixed ambiguity set, and support-wise maximum regret are all unchanged by the transformation.
\end{proposition}

Because the same $c(\bm\xi)$ enters both the state-specific best score and every candidate score, it cancels in $\widetilde L(a,\bm\xi)$. The invariance therefore holds before any expectation or worst-case aggregation is taken. For MP, any scenario-dependent pressure component shared by all feasible phases is removed from the phase regret, which is the structural property used in the comparison with absolute-pressure robustification below. The proof is given in Appendix~\ref{app:regret_structural_proofs}.

\begin{proposition}[Regret versus absolute-score robustification]
\label{prop:robustification_distinction}
Let $\mathfrak P$ be an ambiguity set for which the relevant expectations are finite. If there exists a constant $\bar M$ such that $\mathbb E_{\mathbb P}[M(\bm\xi)]=\bar M$ for every $\mathbb P\in\mathfrak P$, then
\begin{equation}
    J_{\mathrm{reg}}(a;\mathfrak P)=\bar M-J_{\mathrm{abs}}(a;\mathfrak P)
    \qquad\forall a\in\mathcal A,
    \label{eq:robust_equivalence_condition}
\end{equation}
and hence $\arg\min_a J_{\mathrm{reg}}(a;\mathfrak P)=\arg\max_a J_{\mathrm{abs}}(a;\mathfrak P)$. In general, when the expected state-specific best score varies over the ambiguity set, this decision equivalence is not guaranteed: there exist finite score systems and Wasserstein ambiguity sets for which the two criteria select different actions.
\end{proposition}

When $\mathbb E_{\mathbb P}[M(\bm\xi)]$ is constant over $\mathfrak P$, it separates from the supremum and the two robust criteria differ only by the common constant $\bar M$, so their optimizer sets coincide. When this expectation varies across distributions, that separation is no longer available; the robust regret and robust absolute-score criteria can then rank actions differently. Together with Proposition~\ref{prop:common_shift_invariance}, this identifies the structural distinction used to motivate the matched absolute-pressure benchmarks below. Appendix~\ref{app:regret_structural_proofs} gives the proof and the finite Wasserstein example.

For MP, the distributionally robust absolute-pressure counterpart is
\begin{equation}
    \Phi_n^{\mathrm{DR-P}}(\bm s_n,t)
    =
    \inf_{\bm p_n\in\mathcal P_n(\varepsilon_n)}
    \sum_{k=1}^{K_n}p_{n,k}H_{n,k}(\bm s_n,t),
    \label{eq:dr_pressure_score}
\end{equation}
and the support-wise absolute-pressure score is
\begin{equation}
    \Phi_n^{\mathrm{Set-P}}(\bm s_n,t)
    =
    \min_{k\in\mathcal K_n}H_{n,k}(\bm s_n,t).
    \label{eq:set_pressure_score}
\end{equation}
The corresponding DR-Pressure and Set-Pressure controllers maximize $\Phi_n^{\mathrm{DR-P}}$ and $\Phi_n^{\mathrm{Set-P}}$, respectively. In contrast, regret-based robustification removes phase-common scenario pressure components before aggregation and therefore operates on relative phase quality.

\subsection{Relationships among regret formulations}
\label{sec:regret_relationships}

This subsection positions DR Regret between the calibrated stochastic rule and the support-wise worst-case rule.

\begin{remark}[Support-wise maximum regret]
\label{rem:set_regret}
If only the retained support is used, the aggregate score is
\begin{equation}
    G_n^{\mathrm{Set}}(\bm s_n,t)
    =
    \max_{k\in\mathcal K_n}
    L_{n,k}(\bm s_n,t).
    \label{eq:set_regret_score}
\end{equation}
The corresponding controller minimizes the largest phase-selection regret over the retained support and does not use scenario probabilities.
\end{remark}

The stochastic, DR, and support-wise formulations therefore correspond to progressively larger sets of admissible probability distributions:
\(\{\widehat{\bm\pi}_n\}\subseteq\mathcal P_n(\varepsilon_n)\subseteq\Delta_{K_n}.\)
Over the full simplex,
$$\sup_{\bm p_n\in\Delta_{K_n}}\sum_{k=1}^{K_n}p_{n,k}L_{n,k}(\bm s_n,t)=\max_{k\in\mathcal K_n}L_{n,k}(\bm s_n,t)=G_n^{\mathrm{Set}}(\bm s_n,t).$$
The equality holds because an unrestricted probability vector can place all mass on a maximum-regret scenario. Thus, Set Regret is the extreme case in which only the retained support matters and no probability calibration is trusted. Because the associated probability sets are nested, the three regret scores are ordered phase by phase.

\begin{proposition}[Ordering of regret scores]
\label{prop:regret_ordering}
For every candidate phase $\bm s_n$ and every $\varepsilon_n\geq0$,
\begin{equation}
    G_n^{\mathrm{Sto}}(\bm s_n,t)
    \leq
    G_n^{\mathrm{DR}}(\bm s_n,t;\varepsilon_n)
    \leq
    G_n^{\mathrm{Set}}(\bm s_n,t).
    \label{eq:regret_ordering}
\end{equation}
\end{proposition}

The ordering follows from the nested probability sets $\{\widehat{\bm\pi}_n\}\subseteq\mathcal P_n(\varepsilon_n)\subseteq\Delta_{K_n}$: enlarging the admissible set can only increase the worst-case expected regret, while no convex combination can exceed the largest scenario regret. Proposition~\ref{prop:regret_ordering} therefore places DR Regret between the calibrated stochastic score and the support-wise maximum. The next proposition identifies when these two bounds are attained exactly. The proof is given in Appendix~\ref{app:regret_structural_proofs}.

\begin{proposition}[Limiting cases of DR regret]
\label{prop:dr_regret_limits}
For every candidate phase $\bm s_n$, DR Regret MP reduces to Stochastic Regret MP at zero ambiguity radius,
\(G_n^{\mathrm{DR}}(\bm s_n,t;0)=G_n^{\mathrm{Sto}}(\bm s_n,t).\)
Let $\mathcal K_n^{\max}(\bm s_n,t)=\arg\max_{j\in\mathcal K_n}L_{n,j}(\bm s_n,t)$ and define
$\varepsilon_{n,\bm s}^{\mathrm{Set}}(t)=\min_{j\in\mathcal K_n^{\max}(\bm s_n,t)}\sum_{i=1}^{K_n}\widehat\pi_{n,i}d_{n,ij}$.
If $\varepsilon_n\geq\varepsilon_{n,\bm s}^{\mathrm{Set}}(t)$, then
\(G_n^{\mathrm{DR}}(\bm s_n,t;\varepsilon_n)=G_n^{\mathrm{Set}}(\bm s_n,t).\)
\end{proposition}

At $\varepsilon_n=0$, the ambiguity set contains only the calibrated distribution, so the DR score equals the stochastic score. At the other endpoint, $\varepsilon_{n,\bm s}^{\mathrm{Set}}(t)$ is the minimum cost of transporting all calibrated probability mass to a maximum-regret scenario; once this budget is available, that point mass is feasible and the DR score reaches the support-wise maximum. Proposition~\ref{prop:dr_regret_limits} therefore gives the two reference limits for interpreting the Wasserstein radius. These controller relationships are carried into the stability analysis in Section~\ref{sec:theory}. The proof is given in Appendix~\ref{app:regret_structural_proofs}.

\section{Network Stability}
\label{sec:theory}

This section establishes the network-stability properties of Regret MP under uncertain traffic states. The analysis proceeds in three steps. We first express stability through a common Lyapunov-drift bound in which phase pressure is the controllable term. We then quantify how the finite residual representation affects conditional pressure and show that the resulting pressure loss grows at most linearly with queue magnitude. Finally, we compare this queue-linear uncertainty loss with the available net-service reserve to establish strong stability and to compare Regret MP with Point-estimate MP. The extension to fixed positive queue weighting is given in Appendix~\ref{app:weighted_stability}.

\subsection{Stability framework}
\label{sec:stability_preliminaries}

This subsection defines the stability criterion and physical service region, and then derives the drift bound used throughout the analysis. Let $\mathcal G_t$ denote the information available immediately before the phase decision, including the current estimated traffic state, the calibrated residual model, and past controls and observations.

We use strong queue stability as the network-level criterion.
\begin{definition}[Strong queue stability]
\label{def:strong_stability}
The network is strongly stable if
$$
\limsup_{T\rightarrow\infty}
    \frac1T
    \sum_{t=0}^{T-1}
    \mathbb E\|\bm Q(t)\|_1
    <\infty.
$$
\end{definition}

The drift analysis uses the following standard boundedness conditions.
\begin{assumption}[Bounded network primitives]
\label{ass:bounded_primitives}
Let $\mathcal H_t$ be the full pre-arrival sigma-field containing $\mathcal G_t$, the latent queue vector $\bm Q(t)$, the current routing matrix $\bm R(t)$, the service-capacity matrix $\bm C(t)$, and the selected feasible phase $\bm s(t)$. Exogenous arrivals satisfy $\mathbb E[\bm A(t)\mid\mathcal H_t]=\bm\lambda$ and $\mathbb E[\|\bm A(t)\|_2^2\mid\mathcal H_t]\leq\bar A_2<\infty$ almost surely, where $\bm\lambda$ is the mean exogenous-arrival vector and $\bar A_2$ is a finite second-moment bound. Routing matrices are nonnegative and row-substochastic, saturation service satisfies $0<c_m(t)\leq c_m^{\max}<\infty$ almost surely, and $\mathbb E\|\bm Q(0)\|_2^2<\infty$.
\end{assumption}

For any feasible phase $\bm s$, define the potential net-service vector
\(\bm Z(\bm s,t)=[\bm I-\bm R^\top(t)]\bm C(t)\bm s,\)
so that $H(\bm s,t)=\bm Q^\top(t)\bm Z(\bm s,t)$. To characterize the physical service capability of the network, let $\mathfrak U_{\mathrm{phy}}$ be the class of stationary randomized feasible policies whose phase distribution may depend on the current pair $(\bm R(t),\bm C(t))$. For $\mu\in\mathfrak U_{\mathrm{phy}}$, let $\bm s^\mu(t)$ denote the selected phase and $\bm Z^\mu(t)=\bm Z(\bm s^\mu(t),t)$. Define its componentwise long-run net service by $\bar{\bm Z}^{\mu}=\liminf_{T\rightarrow\infty}T^{-1}\sum_{t=0}^{T-1}\mathbb E[\bm Z^\mu(t)]$, where the $\liminf$ is taken componentwise. The physical admissible demand region is
\begin{equation}
    \mathcal D_{\mathrm{phy}}
    =
    \operatorname{cl}
    \left\{
        \bm\lambda\succeq\bm0:
        \exists\,\mu\in\mathfrak U_{\mathrm{phy}}
        \text{ such that }
        \bar{\bm Z}^{\mu}\succeq\bm\lambda
    \right\},
    \label{eq:physical_demand_region}
\end{equation}
and, with $\bm1$ denoting the all-ones vector of compatible dimension, its strict interior is
\begin{equation}
    \mathcal D_{\mathrm{phy}}^{\circ}
    =
    \left\{
    \bm\lambda\succeq\bm0:
    \begin{array}{l}
        \exists\,\mu\in\mathfrak U_{\mathrm{phy}},\ \epsilon_0>0
        \text{ such that}\\[-1pt]
        \bar{\bm Z}^{\mu}\succeq\bm\lambda+\epsilon_0\bm1
    \end{array}
    \right\}.
    \label{eq:physical_demand_interior}
\end{equation}
Here, $\epsilon_0$ is a strict componentwise net-service reserve above demand. For fixed $\bm R$ and $\bm C$, this reduces to the standard MP capacity description $\bm\lambda\preceq[\bm I-\bm R^\top]\bm C\bar{\bm s}$ for some $\bar{\bm s}\in\operatorname{co}(\mathcal S)$ \citep{wongpiromsarn2012distributed,varaiya2013max}.

The following lemma connects queue stability to the phase pressure selected by the controller.
\begin{lemma}[Quadratic drift under potential pressure]
\label{lem:quadratic_drift}
Let $V(\bm Q)=\frac12\|\bm Q\|_2^2$ and $\Delta V(t)=V(\bm Q(t+1))-V(\bm Q(t))$. Under Assumption~\ref{ass:bounded_primitives}, there exists a finite constant $B_0$, independent of the queue state and selected feasible phase, such that
\begin{equation}
    \mathbb E\!\left[
        \Delta V(t)
        \mid\mathcal H_t
    \right]
    \leq
    B_0
    +
    \bm Q^\top(t)\bm\lambda
    -
    H(\bm s(t),t).
    \label{eq:initial_drift_bound}
\end{equation}
\end{lemma}

Assumption~\ref{ass:bounded_primitives} bounds the arrival second moment, service magnitude, and unused-service contribution, which ensures that the constant $B_0$ in Lemma~\ref{lem:quadratic_drift} is finite. By Lemma~\ref{lem:quadratic_drift}, the one-step Lyapunov drift is then bounded by the queue-independent term $B_0$ plus the difference between the arrival term $\bm Q^\top(t)\bm\lambda$ and the pressure $H(\bm s(t),t)$ achieved by the selected phase. Once these network primitives are controlled, the effect of traffic-state uncertainty on stability therefore enters through its effect on the achieved phase pressure. The subsequent analysis quantifies this pressure loss and relates it to queue magnitude. The derivation of Lemma~\ref{lem:quadratic_drift} is given in Appendix~\ref{app:network_stability_proofs}.

When the traffic state is known exactly, the regret representation recovers the standard true-state MP benchmark.
\begin{corollary}[Perfect-information benchmark]
\label{cor:perfect_information_benchmark}
\label{prop:perfect_information}
\label{thm:perfect_information_stability}
Suppose the local uncertainty at every intersection collapses to the current latent traffic state. Then Stochastic Regret MP, DR Regret MP, and Set Regret select a true-state MP phase. If, in addition, Assumption~\ref{ass:bounded_primitives} holds, $(\bm R(t),\bm C(t))$ is exogenous and i.i.d. over time with its current realization observed before phase selection, and $\bm\lambda\in\mathcal D_{\mathrm{phy}}^{\circ}$ with reserve $\epsilon_0>0$, then the resulting controller strongly stabilizes the network and
\begin{equation}
    \limsup_{T\rightarrow\infty}
    \frac1T
    \sum_{t=0}^{T-1}
    \mathbb E\|\bm Q(t)\|_1
    \leq
    \frac{B_0}{\epsilon_0}.
    \label{eq:perfect_information_queue_bound}
\end{equation}
\end{corollary}

With a single latent state, phase regret differs from negative phase pressure only by a phase-independent term, so the regret-based phase choice coincides with true-state MP. Assumption~\ref{ass:bounded_primitives} and Lemma~\ref{lem:quadratic_drift} provide the finite drift constant $B_0$. Membership of $\bm\lambda$ in $\mathcal D_{\mathrm{phy}}^{\circ}$ supplies a comparison policy with reserve $\epsilon_0$, and the exogenous i.i.d. network-state condition converts this long-run reserve into a one-step expected pressure margin. Since true-state MP maximizes realized pressure, its pressure is no smaller than that of the comparison policy; inserting this margin into Lemma~\ref{lem:quadratic_drift} gives the negative queue term $-\epsilon_0\|\bm Q(t)\|_1$ and the bound in~\eqref{eq:perfect_information_queue_bound}. The proof of Corollary~\ref{cor:perfect_information_benchmark} is given in Appendix~\ref{app:network_stability_proofs}.

\subsection{Conditional pressure performance under uncertainty}
\label{sec:network_stability}
\label{sec:true_regret_pressure}

This subsection quantifies how the finite residual representation affects the conditional pressure achieved by Regret MP. The argument first links conditional regret to conditional pressure and then represents scenario-reduction and probability errors on the same residual-space scale.

The retained scenarios approximate a latent residual $\bm\xi_n(t)\in\Xi_n$ that may take values throughout the continuous residual support. For an arbitrary residual $\bm\xi\in\Xi_n$, let $\bm x_n(\bm\xi,t)=\mathcal T_n(\widehat{\bm x}_n(t),\bm\xi)$ denote the reconstructed local state and let $H_n(\bm s_n,\bm\xi,t)$ be the corresponding phase pressure. Define the continuous-state phase regret by
\begin{equation}
    L_n(\bm s_n,\bm\xi,t)
    =
    \max_{\bm u_n\in\mathcal S_n}
    H_n(\bm u_n,\bm\xi,t)
    -
    H_n(\bm s_n,\bm\xi,t).
    \label{eq:continuous_regret}
\end{equation}
Given the decision information $\mathcal G_t$, define the conditional expected pressure
$\bar H_n(\bm s_n,t)=\mathbb E[H_n(\bm s_n,\bm\xi_n(t),t)\mid\mathcal G_t]$
and the corresponding conditional expected regret
$\bar L_n(\bm s_n,t)=\mathbb E[L_n(\bm s_n,\bm\xi_n(t),t)\mid\mathcal G_t]$.

\begin{lemma}[Conditional regret--pressure identity]
\label{lem:conditional_regret_pressure}
For every candidate phase,
$$
\bar L_n(\bm s_n,t)
=
\mathbb E\!\left[
    \max_{\bm u_n\in\mathcal S_n}
    H_n(\bm u_n,\bm\xi_n(t),t)
    \mid\mathcal G_t
\right]
-
\bar H_n(\bm s_n,t).
$$
Consequently, minimizing conditional expected regret is equivalent to maximizing conditional expected pressure. If $\bm s_n^\star(t)\in\arg\max_{\bm s_n\in\mathcal S_n}\bar H_n(\bm s_n,t)$, then
$$
\bar H_n(\bm s_n^\star,t)-\bar H_n(\bm s_n,t)
=
\bar L_n(\bm s_n,t)-\bar L_n(\bm s_n^\star,t).
$$
\end{lemma}

Lemma~\ref{lem:conditional_regret_pressure} follows by taking the conditional expectation of the continuous-state regret definition in~\eqref{eq:continuous_regret}; the conditional expectation of the state-wise maximum is common to all candidate phases and therefore cancels in phase comparisons. As a result, any bound on the difference between an implemented aggregate regret score and the true conditional regret can be translated directly into a bound on conditional pressure loss. This identity is used with the finite-scenario approximation bounds to obtain Proposition~\ref{prop:local_approx_pressure}. The proof of the lemma is given in Appendix~\ref{app:network_stability_proofs}.

To quantify the finite representation, let $\iota_n:\Xi_n\to\mathcal K_n$ map each residual to the index of its nearest retained scenario. Let
$\pi_{n,k}(t)=\Pr\{\iota_n(\bm\xi_n(t))=k\mid\mathcal G_t\}$
denote the current conditional probability of representative scenario $k$, and let $\bm\pi_n(t)=(\pi_{n,1}(t),\ldots,\pi_{n,K_n}(t))$. The sensitivity of phase regret to residual perturbations is measured by
\begin{equation}
    \Lambda_n(t)
    =
    \sup_{\bm s_n\in\mathcal S_n}
    \sup_{\substack{\bm\xi,\bm\zeta\in\Xi_n\\ \bm\xi\neq\bm\zeta}}
    \frac{
        |L_n(\bm s_n,\bm\xi,t)-L_n(\bm s_n,\bm\zeta,t)|
    }{
        d_n(\bm\xi,\bm\zeta)
    },
    \label{eq:regret_lipschitz_modulus}
\end{equation}
where $d_n$ is the local residual-space metric defined in Section~\ref{sec:dr_regret}. Let $d_n^{\max}=\max_{i,j\in\mathcal K_n}d_{n,ij}$ denote the diameter of the retained scenario support.

\begin{assumption}[Scenario approximation and probability coverage]
\label{ass:coverage}
For each intersection $n$, the residual support $\Xi_n$ is compact and the reconstruction map is well defined on $\Xi_n$. There exist finite constants $\tau_n\geq0$ and $\varepsilon_n^{\mathrm{cov}}\geq0$ such that
\begin{equation}
    \mathbb E\!\left[
        d_n\!\left(
            \bm\xi_n(t),
            \bar{\bm\xi}_n^{\iota_n(\bm\xi_n(t))}
        \right)
        \mid\mathcal G_t
    \right]
    \leq\tau_n,
    \label{eq:quantization_distortion_bound}
\end{equation}
and
\begin{equation}
    W_n\!\left(
        \bm\pi_n(t),
        \widehat{\bm\pi}_n
    \right)
    \leq\varepsilon_n^{\mathrm{cov}}
    \qquad\text{a.s.}
    \label{eq:coverage_condition}
\end{equation}
For DR Regret MP, the ambiguity radius satisfies $\varepsilon_n\geq\varepsilon_n^{\mathrm{cov}}$. For Set Regret, only~\eqref{eq:quantization_distortion_bound} and a finite support diameter are required.
\end{assumption}

The constant $\tau_n$ bounds the average distortion introduced by scenario reduction, while $\varepsilon_n^{\mathrm{cov}}$ bounds the probability mismatch between the current representative distribution and the calibrated distribution. The detailed regret-approximation bounds implied by this assumption are given in Appendix~\ref{app:network_stability_proofs}. Combining them with Lemma~\ref{lem:conditional_regret_pressure} gives the main local pressure result.

\begin{proposition}[Approximate conditional maximum pressure]
\label{prop:local_approx_pressure}
Let $\bm s_n^\star(t)\in\arg\max_{\bm s_n\in\mathcal S_n}\bar H_n(\bm s_n,t)$. Under Assumption~\ref{ass:coverage}, the phase $\bm s_n^\psi(t)$ selected by regret aggregation $\psi$ satisfies
\begin{equation}
    \bar H_n(\bm s_n^\psi(t),t)
    \geq
    \bar H_n(\bm s_n^\star(t),t)
    -
    a_n^\psi\Lambda_n(t),
    \label{eq:local_approx_pressure}
\end{equation}
where the uncertainty coefficient is
\begin{equation}
    a_n^\psi
    =
    \begin{cases}
        2(\varepsilon_n^{\mathrm{cov}}+\tau_n),
        & \psi=\mathrm{Sto},\\[1mm]
        \varepsilon_n+\varepsilon_n^{\mathrm{cov}}+2\tau_n,
        & \psi=\mathrm{DR},\\[1mm]
        d_n^{\max}+2\tau_n,
        & \psi=\mathrm{Set}.
    \end{cases}
    \label{eq:aggregation_pressure_coefficient}
\end{equation}
\end{proposition}

Assumption~\ref{ass:coverage} controls the two approximation steps introduced by the finite residual representation: $\tau_n$ bounds the residual-to-representative distortion, while $\varepsilon_n^{\mathrm{cov}}$ bounds the mismatch between the current and calibrated representative probabilities; for DR Regret MP, $\varepsilon_n\geq\varepsilon_n^{\mathrm{cov}}$ also makes the current representative distribution feasible in the ambiguity set. These quantities yield controller-specific bounds between the implemented aggregate regret and the true conditional regret. Lemma~\ref{lem:conditional_regret_pressure} then converts the resulting regret gap into exactly the same conditional-pressure gap, producing the coefficient $a_n^\psi$ in~\eqref{eq:aggregation_pressure_coefficient}. Proposition~\ref{prop:local_approx_pressure} therefore supplies the local pressure-loss term $a_n^\psi\Lambda_n(t)$ that is carried into the network-level stability analysis. The finite-scenario bounds and the proof are given in Appendix~\ref{app:network_stability_proofs}.

\subsection{Queue-linear pressure loss and strong stability}
\label{sec:queue_linear_stability}

The pressure gap in~\eqref{eq:local_approx_pressure} still depends on the state-dependent sensitivity $\Lambda_n(t)$. This subsection bounds that sensitivity by an affine function of queue magnitude and then compares the resulting queue-proportional pressure loss with the available service reserve.

Let $\Xi_n^{0}=\Xi_n\cup\{\bm0\}$ and define the local pressure modulus
$$
\Upsilon_n(t)
    =
    \sup_{\bm s_n\in\mathcal S_n}
    \sup_{\substack{\bm\xi,\bm\zeta\in\Xi_n^{0}\\ \bm\xi\neq\bm\zeta}}
    \frac{
        |H_n(\bm s_n,\bm\xi,t)-H_n(\bm s_n,\bm\zeta,t)|
    }{
        d_n(\bm\xi,\bm\zeta)
    }.
$$
Here, $\Upsilon_n(t)$ measures the sensitivity of phase pressure to residual variation.

\begin{assumption}[Affine reconstruction sensitivity]
\label{ass:affine_reconstruction}
For every intersection $n$, there exist finite constants $\beta_n\geq0$, $h_n>0$, $b_{Q,n}\geq0$, and $\eta_{Q,n}>0$, independent of the queue state, such that
\begin{equation}
    \Upsilon_n(t)
    \leq
    \beta_n
    +
    h_n\|\widehat{\bm Q}_{\mathcal L_n}(t)\|_1,
    \label{eq:affine_pressure_sensitivity}
\end{equation}
and
\begin{equation}
    \|\widehat{\bm Q}_{\mathcal L_n}(t)\|_1
    \leq
    b_{Q,n}
    +
    \eta_{Q,n}\|\bm Q_{\mathcal L_n}(t)\|_1.
    \label{eq:affine_estimator_growth}
\end{equation}
\end{assumption}

The constants $\beta_n$ and $b_{Q,n}$ collect queue-independent terms, while $h_n$ and $\eta_{Q,n}$ determine how pressure sensitivity grows with queue magnitude. The multiplicative reconstruction in Section~\ref{sec:uncertain_state} and the bounded additive reconstruction used in the microscopic experiments satisfy this form; the verification is given in Appendix~\ref{app:network_stability_proofs}.

\begin{proposition}[Queue-linear pressure-loss bound]
\label{prop:queue_linear_pressure_loss}
\label{lem:queue_linear_regret}
Under Assumption~\ref{ass:affine_reconstruction}, define
\begin{equation}
    \gamma_n
    =
    2(\beta_n+h_nb_{Q,n}),
    \qquad
    \ell_n
    =
    2h_n\eta_{Q,n},
    \label{eq:affine_sensitivity_constants}
\end{equation}
where $\gamma_n$ is the queue-independent local sensitivity bound and $\ell_n$ is its queue-linear coefficient. Then
\begin{equation}
    \Lambda_n(t)
    \leq
    2\Upsilon_n(t)
    \leq
    \gamma_n
    +
    \ell_n\|\bm Q_{\mathcal L_n}(t)\|_1.
    \label{eq:queue_linear_regret}
\end{equation}
For $\psi\in\{\mathrm{Sto},\mathrm{DR},\mathrm{Set}\}$, define
$$
B_\xi^\psi=\sum_{n\in\mathcal N}a_n^\psi\gamma_n,
\qquad
\delta^\psi=\max_{m\in\mathcal M}\sum_{n:m\in\mathcal L_n}a_n^\psi\ell_n.
$$
Here, $B_\xi^\psi$ is the bounded network-level uncertainty-loss term and $\delta^\psi$ is the coefficient of the queue-proportional uncertainty loss. The network-wide pressure gap satisfies
\begin{equation}
    \sum_{n\in\mathcal N}
    a_n^\psi\Lambda_n(t)
    \leq
    B_\xi^\psi
    +
    \delta^\psi\|\bm Q(t)\|_1.
    \label{eq:global_pressure_loss_bound}
\end{equation}
Since the left-hand side is $\mathcal G_t$-measurable, it also follows that
\begin{equation}
    \sum_{n\in\mathcal N}
    a_n^\psi\Lambda_n(t)
    \leq
    B_\xi^\psi
    +
    \delta^\psi
    \mathbb E[\|\bm Q(t)\|_1\mid\mathcal G_t].
    \label{eq:conditional_global_pressure_loss_bound}
\end{equation}
\end{proposition}

Assumption~\ref{ass:affine_reconstruction} first bounds the phase-pressure modulus $\Upsilon_n(t)$ by an affine function of the estimated queues and then bounds those estimates by an affine function of the latent queues. Because phase regret is the difference between a state-wise maximum pressure and a candidate pressure, its modulus satisfies $\Lambda_n(t)\leq2\Upsilon_n(t)$, which gives the finite local constants $\gamma_n$ and $\ell_n$ in~\eqref{eq:queue_linear_regret}. Combining this bound with the finite coefficients $a_n^\psi$ from Proposition~\ref{prop:local_approx_pressure} and summing over the finite network yields the finite network terms $B_\xi^\psi$ and $\delta^\psi$. Proposition~\ref{prop:queue_linear_pressure_loss} therefore converts the state-dependent pressure gap into the drift-compatible form $B_\xi^\psi+\delta^\psi\|\bm Q\|_1$, with $\delta^\psi$ becoming the queue-proportional uncertainty coefficient used in the stability theorem. The proof is given in Appendix~\ref{app:network_stability_proofs}.

To state the stability result, let $\mathcal G_t^Q=\sigma(\mathcal G_t,\bm Q(t))$. Let $\widetilde{\bm s}(t)$ be generated by a stationary randomized feasible comparison policy whose phase randomization is based on $\mathcal G_t$ and is conditionally independent of the latent traffic state given $\mathcal G_t$. Write $\bm Z^{\widetilde s}(t)=\bm Z(\widetilde{\bm s}(t),t)$ for its potential net-service vector.

\begin{definition}[Provable stability region]
\label{def:provable_stability_region}
For $\psi\in\{\mathrm{Sto},\mathrm{DR},\mathrm{Set}\}$, define $\mathcal D_{\mathrm{prov}}^\psi$ as the set of $\bm\lambda\in\mathcal D_{\mathrm{phy}}$ for which there exist such a comparison policy and a constant $\epsilon_0>\delta^\psi$ satisfying
\begin{equation}
    \mathbb E\!\left[
        \bm Z^{\widetilde s}(t)
        \mid\mathcal G_t^Q
    \right]
    \succeq
    \bm\lambda+\epsilon_0\bm1
    \qquad\text{a.s. for every }t.
    \label{eq:conditional_service_reserve}
\end{equation}
\end{definition}

Thus, $\epsilon_0$ is the conditional net-service reserve available to offset the queue-proportional uncertainty loss $\delta^\psi$.

\begin{theorem}[Strong stability of regret-based MP]
\label{thm:regret_mp_stability}
Suppose Assumptions~\ref{ass:bounded_primitives}, \ref{ass:coverage}, and~\ref{ass:affine_reconstruction} hold, with the probability-coverage requirement omitted for Set Regret. If $\bm\lambda\in\mathcal D_{\mathrm{prov}}^\psi$ with reserve $\epsilon_0$, then the corresponding regret-based MP controller strongly stabilizes the network and
\begin{equation}
    \limsup_{T\rightarrow\infty}
    \frac1T
    \sum_{t=0}^{T-1}
    \mathbb E\|\bm Q(t)\|_1
    \leq
    \frac{B_0+B_\xi^\psi}{\epsilon_0-\delta^\psi}.
    \label{eq:average_queue_bound}
\end{equation}
\end{theorem}

Theorem~\ref{thm:regret_mp_stability} combines the preceding bounds in the order in which they enter the drift. Assumption~\ref{ass:bounded_primitives} and Lemma~\ref{lem:quadratic_drift} make the queue-independent drift term $B_0$ finite. Assumption~\ref{ass:coverage}, together with Lemma~\ref{lem:conditional_regret_pressure}, gives the finite local pressure-loss coefficient $a_n^\psi$ and the conditional-pressure gap in Proposition~\ref{prop:local_approx_pressure}. Assumption~\ref{ass:affine_reconstruction} and Proposition~\ref{prop:queue_linear_pressure_loss} then make $\gamma_n$ and $\ell_n$ finite and aggregate the local losses into the finite network terms $B_\xi^\psi$ and $\delta^\psi$. Membership in $\mathcal D_{\mathrm{prov}}^\psi$ supplies a comparison policy whose conditional net-service reserve $\epsilon_0$ exceeds $\delta^\psi$. Comparing the selected phase with this policy and substituting the resulting pressure bound into Lemma~\ref{lem:quadratic_drift} gives
$$
\mathbb E[\Delta V(t)]
\leq
B_0+B_\xi^\psi
-
(\epsilon_0-\delta^\psi)\mathbb E\|\bm Q(t)\|_1.
$$
The coefficient of the queue term is strictly negative because $\epsilon_0>\delta^\psi$; summing this drift inequality over time therefore gives strong stability and the average-queue bound in~\eqref{eq:average_queue_bound}. The complete proof of Theorem~\ref{thm:regret_mp_stability} is given in Appendix~\ref{app:network_stability_proofs}.

\subsection{Comparison with Point-estimate MP}
\label{sec:point_stability_comparison}

The preceding result characterizes Regret MP through the coefficient $\delta^\psi$. To compare it with Point-estimate MP on the same pressure-loss scale, define the conditional distance from the latent residual to the zero-residual point estimate by
\begin{equation}
    \rho_n^{\mathrm P}(t)
    =
    \mathbb E\!\left[
        d_n(\bm\xi_n(t),\bm0)
        \mid\mathcal G_t
    \right],
    \qquad
    \rho_n^{\mathrm P}(t)\leq\bar\rho_n^{\mathrm P}<\infty,
    \label{eq:point_residual_distance}
\end{equation}
where $\bar\rho_n^{\mathrm P}$ is a uniform upper bound on the conditional point-estimation distance. Define
$$
B_\xi^{\mathrm P}=\sum_{n\in\mathcal N}\bar\rho_n^{\mathrm P}\gamma_n,
\qquad
\delta^{\mathrm P}=\max_{m\in\mathcal M}\sum_{n:m\in\mathcal L_n}\bar\rho_n^{\mathrm P}\ell_n.
$$
Here, $B_\xi^{\mathrm P}$ and $\delta^{\mathrm P}$ are the bounded and queue-proportional Point-estimate pressure-loss coefficients, respectively. The same pressure-sensitivity argument gives the network-wide conditional pressure-loss bound
$$
\begin{array}{c}
\displaystyle
\sum_{n\in\mathcal N}
\left[
\max_{\bm s_n\in\mathcal S_n}\bar H_n(\bm s_n,t)
-
\bar H_n(\bm s_n^{\mathrm P}(t),t)
\right]
\\[2pt]
\displaystyle
\leq
B_\xi^{\mathrm P}
+
\delta^{\mathrm P}\mathbb E[\|\bm Q(t)\|_1\mid\mathcal G_t].
\end{array}
$$
The uniform bound $\bar\rho_n^{\mathrm P}$ and Assumption~\ref{ass:affine_reconstruction} make $B_\xi^{\mathrm P}$ and $\delta^{\mathrm P}$ finite. Let $\mathcal D_{\mathrm{prov}}^{\mathrm P}$ be the set of $\bm\lambda\in\mathcal D_{\mathrm{phy}}$ for which a comparison policy satisfies~\eqref{eq:conditional_service_reserve} with some reserve $\epsilon_0>\delta^{\mathrm P}$. Assumption~\ref{ass:bounded_primitives} and Lemma~\ref{lem:quadratic_drift} then provide the same finite drift term $B_0$ as for Regret MP, while the Point-estimate pressure-loss bound contributes $B_\xi^{\mathrm P}+\delta^{\mathrm P}\mathbb E[\|\bm Q(t)\|_1\mid\mathcal G_t]$. The resulting drift has queue-linear coefficient $-(\epsilon_0-\delta^{\mathrm P})$, so Point-estimate MP strongly stabilizes every demand in $\mathcal D_{\mathrm{prov}}^{\mathrm P}$. Appendix~\ref{app:network_stability_proofs} gives the derivation.

The queue-linear coefficients now provide a direct comparison of the sufficient stability conditions.
\begin{corollary}[Comparison of provable stability regions]
\label{cor:stability_region_comparison}
If
$
2(\varepsilon_n^{\mathrm{cov}}+\tau_n)
\leq
\bar\rho_n^{\mathrm P}
\ \text{for all }n,
$
then
$
\mathcal D_{\mathrm{prov}}^{\mathrm P}
\subseteq
\mathcal D_{\mathrm{prov}}^{\mathrm{Sto}}.
$
If
$
\varepsilon_n+\varepsilon_n^{\mathrm{cov}}+2\tau_n
\leq
\bar\rho_n^{\mathrm P}
\ \text{for all }n,
$
then
$
\mathcal D_{\mathrm{prov}}^{\mathrm P}
\subseteq
\mathcal D_{\mathrm{prov}}^{\mathrm{DR}}
\subseteq
\mathcal D_{\mathrm{prov}}^{\mathrm{Sto}}.
$
\end{corollary}

For Stochastic Regret MP, the first condition gives $a_n^{\mathrm{Sto}}\leq\bar\rho_n^{\mathrm P}$ at every intersection. Because $\ell_n\geq0$ and the same local neighborhoods enter the definitions of the network coefficients, this local ordering implies $\delta^{\mathrm{Sto}}\leq\delta^{\mathrm P}$. The second condition similarly gives $a_n^{\mathrm{DR}}\leq\bar\rho_n^{\mathrm P}$ and hence $\delta^{\mathrm{DR}}\leq\delta^{\mathrm P}$. Moreover, $\varepsilon_n\geq\varepsilon_n^{\mathrm{cov}}$ implies $a_n^{\mathrm{Sto}}\leq a_n^{\mathrm{DR}}$ and therefore $\delta^{\mathrm{Sto}}\leq\delta^{\mathrm{DR}}$. Since all provable regions use the same comparison-policy reserve condition and differ through these queue-linear coefficients, the coefficient orderings give the inclusions in Corollary~\ref{cor:stability_region_comparison}. The proof is given in Appendix~\ref{app:network_stability_proofs}; the fixed-positive-weight extension, including link-length normalization, is given in Appendix~\ref{app:weighted_stability}.

\section{Numerical Experiments}
\label{sec:numerical}
The numerical experiments evaluate regret-based phase selection and isolate the effects of traffic-state error, probability mismatch, and the robustification target.

\subsection{Experimental setting}
\label{sec:numerical_setting}

The numerical experiments use the eight-intersection corridor in Fig.~\ref{fig:numerical_topology}, which also summarizes the baseline demand, service, routing, and phase settings. Demand alternates between off-peak and peak periods, with Poisson arrivals and stochastic variations in movement service and routing. Each run contains 5000 control steps, with the first 500 discarded as burn-in; common traffic random numbers are used across controllers.

\begin{figure}
\centering
\includegraphics[width=0.8\linewidth]{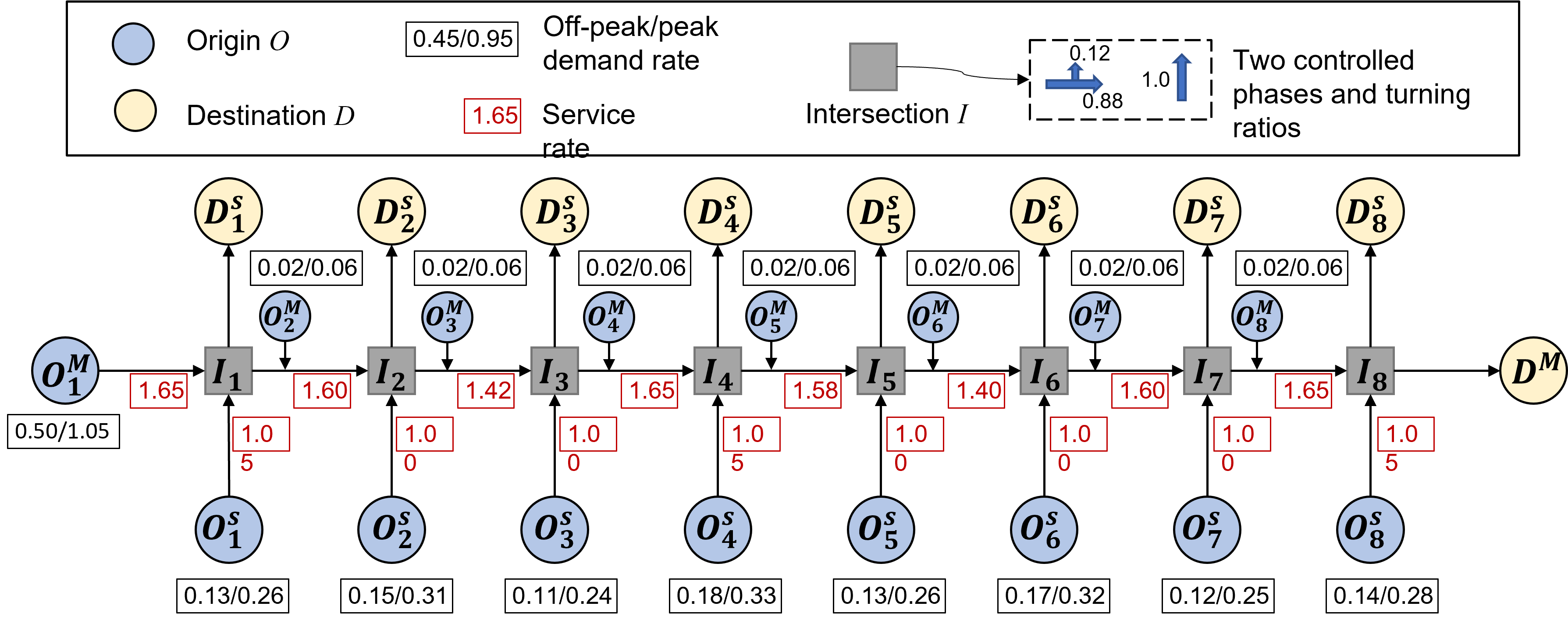}
\caption{Baseline topology and traffic parameters for the numerical experiments. Blue circles denote origins, yellow circles destinations, and gray squares signalized intersections. The mainline and side-street movements are served by separate phases at each intersection; the figure reports the exogenous demands, nominal movement saturation-flow rates, and mainline downstream routing proportions used in the baseline setting.}
\label{fig:numerical_topology}
\end{figure}

Traffic-state uncertainty jointly affects queues, routing proportions, and movement service. Historical residuals are generated from four error-severity modes with persistent temporal variation, and each intersection uses $N_{\mathrm{cal}}=1500$ residual samples reduced to $K_{\mathrm{sc}}=8$ representative joint scenarios. Stochastic Regret MP uses the calibrated scenario probabilities, while DR Regret MP uses the corresponding bootstrap-calibrated Wasserstein ambiguity set. Unless otherwise stated, the online residual distribution matches the historical calibration. Results are averaged over 20 independently generated error environments and two common traffic seeds per environment.

The main comparison includes True-state MP, Point-estimate MP, Stochastic Regret MP, and DR Regret MP. True-state MP provides the oracle reference and Point-estimate MP the practical benchmark. Set Regret is used as the support-wise reference, while DR-Pressure and Set-Pressure are included only in the formulation comparison to isolate the effect of robustifying phase regret rather than absolute pressure.

The experiments examine three questions. The baseline and formulation tests evaluate overall performance and compare regret with absolute-pressure robustification. The uncertainty-sensitivity tests vary residual magnitude and probability mismatch. The phase-decision analysis examines how regret margins and the Wasserstein radius affect implemented phase choices.

The principal network-level metric is the mean total queue. Decision quality is evaluated by agreement with True-state MP and by the normalized true-state pressure regret
$$
\mathrm{NPR}
    =
    \frac{1}{T}
    \sum_{t=1}^{T}
    \frac{
        \sum_{n\in\mathcal N}
        \left[
            H_n^{\max}(t)-H_n(\bm s_n(t),t)
        \right]
    }
    {1+\|\bm Q(t)\|_1}.
$$
The numerator is the true-state pressure loss of the implemented phases, while the denominator normalizes this loss by network load. Scenario best-phase disagreement and controller disagreement are additionally reported to distinguish variation in plausible phase rankings from changes in the implemented decision.
\subsection{Results and analysis}
\label{sec:numerical_results}

Figure~\ref{fig:numerical_main} summarizes baseline performance and robustification choice (A--D), sensitivity to residual magnitude and probability shift (E--G), and phase-decision sensitivity to regret margin and DR radius (H--I).

\begin{figure}
\centering
\includegraphics[width=0.95\linewidth]{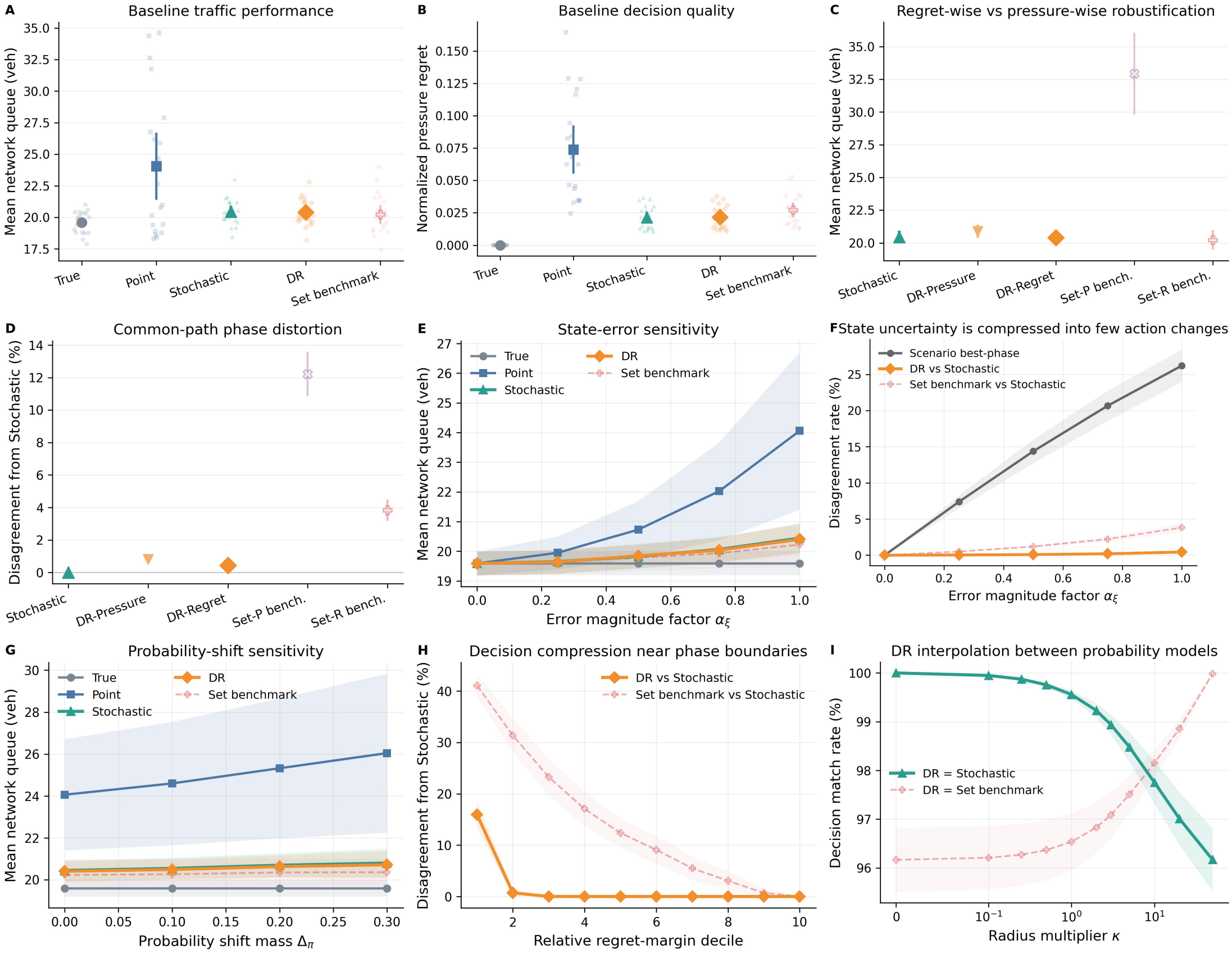}
\caption{Numerical evaluation of Stochastic and DR Regret MP under traffic-state uncertainty; Set Regret is retained as the support-wise benchmark. (A) Baseline mean network queue. (B) Baseline normalized true-state pressure regret. (C) Mean network queue when uncertainty is applied to absolute phase pressure or to phase regret. (D) Common-state phase disagreement from Stochastic Regret MP for the same pressure-wise and regret-wise rules. (E) Mean network queue as the residual magnitude is scaled by $\alpha_\xi$. (F) Scenario best-phase disagreement and controller-level disagreement over the same error-magnitude sweep. (G) Mean network queue under increasing probability shift $\Delta_\pi$. (H) Phase disagreement from Stochastic Regret MP conditioned on the stochastic regret-margin decile. (I) Agreement of DR Regret MP with the stochastic endpoint and the support-wise benchmark as the Wasserstein-radius multiplier $\kappa$ increases. Error bars and shaded bands show 95\% confidence intervals across independent error environments where applicable.}
\label{fig:numerical_main}
\end{figure}

Panels~A--B compare the four main controllers and show that the two regret-based methods recover most of the performance lost when MP acts directly on an uncertain point estimate. Relative to Point-estimate MP, Stochastic Regret MP and DR Regret MP reduce the mean network queue by about 15\% and the normalized true-state pressure regret by about 70\%. Their nearly identical baseline performance is itself informative: when the calibrated scenario probabilities remain representative, Stochastic Regret MP already captures most of the available benefit and DR Regret MP has little reason to alter the nominal decision.

Panels~C--D isolate the \emph{object} being robustified. DR Regret MP yields lower queues than DR-Pressure even though both controllers use the same latent-state scenarios, nominal probabilities, scenario metric, and ambiguity radius; the support-wise comparison shows the same pattern more strongly. The improvement therefore cannot be attributed to more uncertainty information or a different level of conservatism. Instead, it supports Propositions~\ref{prop:common_shift_invariance} and~\ref{prop:robustification_distinction}: regret-based robustification removes scenario-wise pressure components shared by competing phases and concentrates robustness on score differences that can change the phase ranking.

Panels~E--G separate uncertainty in the traffic state from uncertainty in its probability distribution. As residual magnitude increases, Point-estimate MP deteriorates substantially faster than Stochastic Regret MP and DR Regret MP. At the largest residual magnitude, the scenario-best phase varies across plausible latent states in about 26\% of cases, yet the implemented decisions of Stochastic Regret MP and DR Regret MP differ in less than 1\% of cases. This gap shows that large state uncertainty does not automatically imply large control uncertainty: much of the variation is filtered out by the discrete phase ranking. When the online scenario probabilities are shifted away from the calibration distribution, DR Regret MP degrades slightly less than Stochastic Regret MP, consistent with the intended role of Wasserstein protection against probability mismatch.

The remaining question is when changes in regret aggregation are large enough to change the selected phase. Let $F_n$ be a reference aggregate-regret score with unique minimizing phase $\bm s_n^0$, let $F_{n,(1)}$ and $F_{n,(2)}$ be its smallest and second-smallest phase scores, and define the regret margin $\Gamma_n=F_{n,(2)}-F_{n,(1)}$. The following margin condition gives a sufficient criterion for the phase decision to remain unchanged.

\begin{proposition}[Phase-decision invariance]
\label{prop:decision_invariance}
Let $\widetilde F_n$ be another phase score satisfying $\max_{\bm s_n\in\mathcal S_n}|\widetilde F_n(\bm s_n)-F_n(\bm s_n)|\leq\vartheta_n$ for some $\vartheta_n\geq0$. If
\(\Gamma_n>2\vartheta_n,\)
then $\bm s_n^0$ is also the unique minimizer of $\widetilde F_n$.
\end{proposition}

The uniform perturbation bound limits the score change of both the reference phase and any competing phase to at most $\vartheta_n$. Their pairwise gap can therefore decrease by at most $2\vartheta_n$. Proposition~\ref{prop:decision_invariance} shows that a reference margin $\Gamma_n>2\vartheta_n$ is sufficient to preserve the unique minimizing phase, which links score perturbations directly to implemented phase changes in the numerical analysis below. The proof is given in Appendix~\ref{app:decision_invariance_proof}.

Panels~H--I confirm this decision-compression mechanism. Stochastic--DR disagreement is concentrated almost entirely in the lowest regret-margin decisions and becomes negligible once the leading phase has a clear score advantage. As the Wasserstein-radius multiplier increases, DR Regret MP moves continuously from the stochastic endpoint toward the support-wise benchmark, yet the mean network queue changes by only about 1\% over the full sweep. Together, these results explain why different uncertainty models can produce materially different regret scores while changing relatively few implemented signal actions: distributional robustness matters primarily near phase-decision boundaries.

\section{Microscopic Simulation Evaluation}
\label{sec:microscopic}
\label{sec:simulation}

The SUMO experiments test whether the phase-selection improvements observed in Section~\ref{sec:numerical} translate into closed-loop traffic benefits under detector-based state estimation.
The microscopic experiments use SUMO 1.26.0 with TraCI on the 28-intersection Lianyungang network shown in Fig.~\ref{fig:simulation_scenario}. The 7200~s time-varying demand is shown in Fig.~\ref{fig:simulation_scenario}(b), and simulations may continue to 9000~s for network clearance. SUMO runs at 1~s resolution, while detector aggregation and phase decisions occur every 15~s. A phase change uses 3~s yellow followed by 12~s of the selected phase; an unchanged phase continues for the full interval.

\begin{figure}
\centering
\begin{tabular}{@{}cc@{}}
\includegraphics[width=0.47\linewidth]{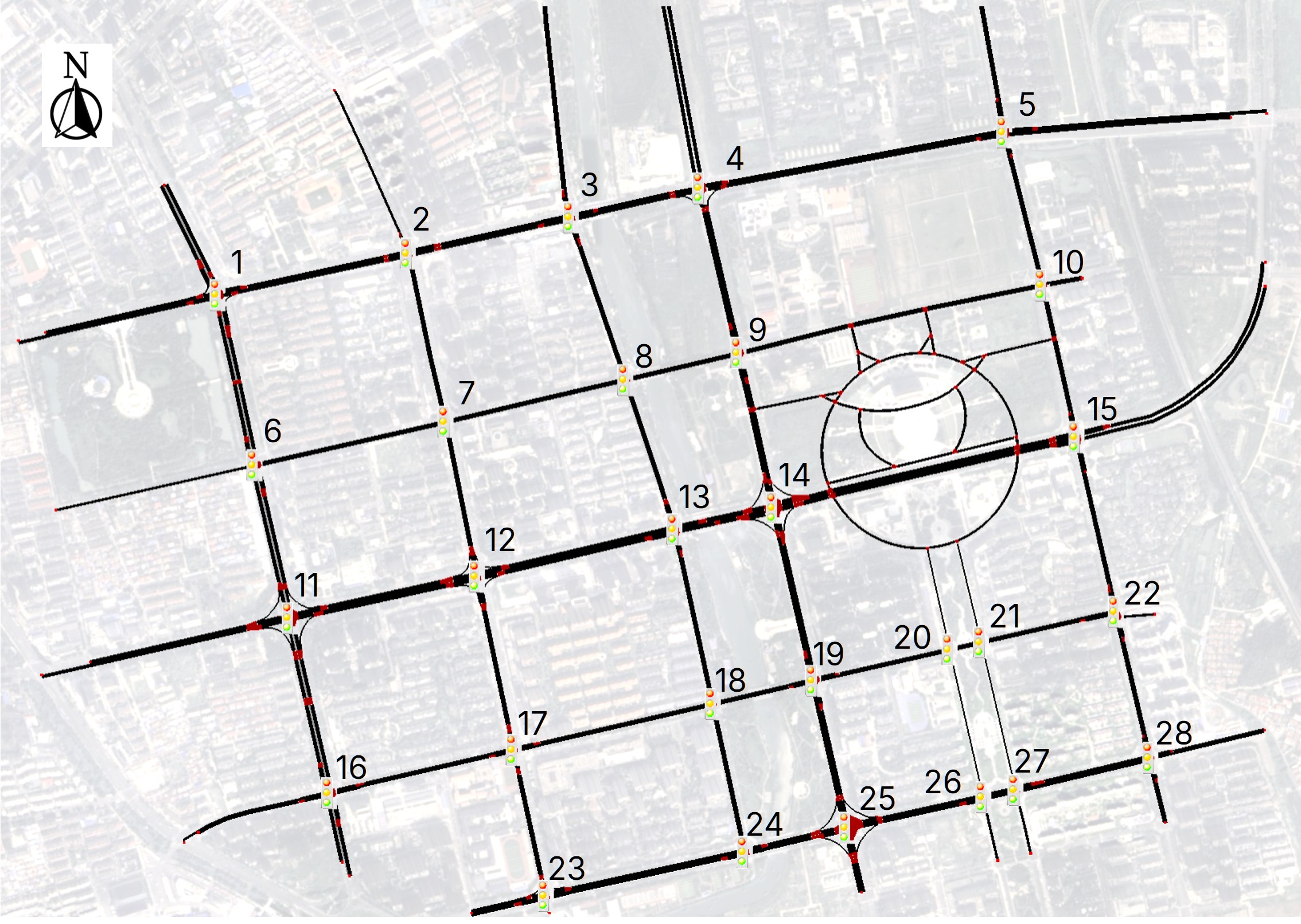} &
\includegraphics[width=0.47\linewidth]{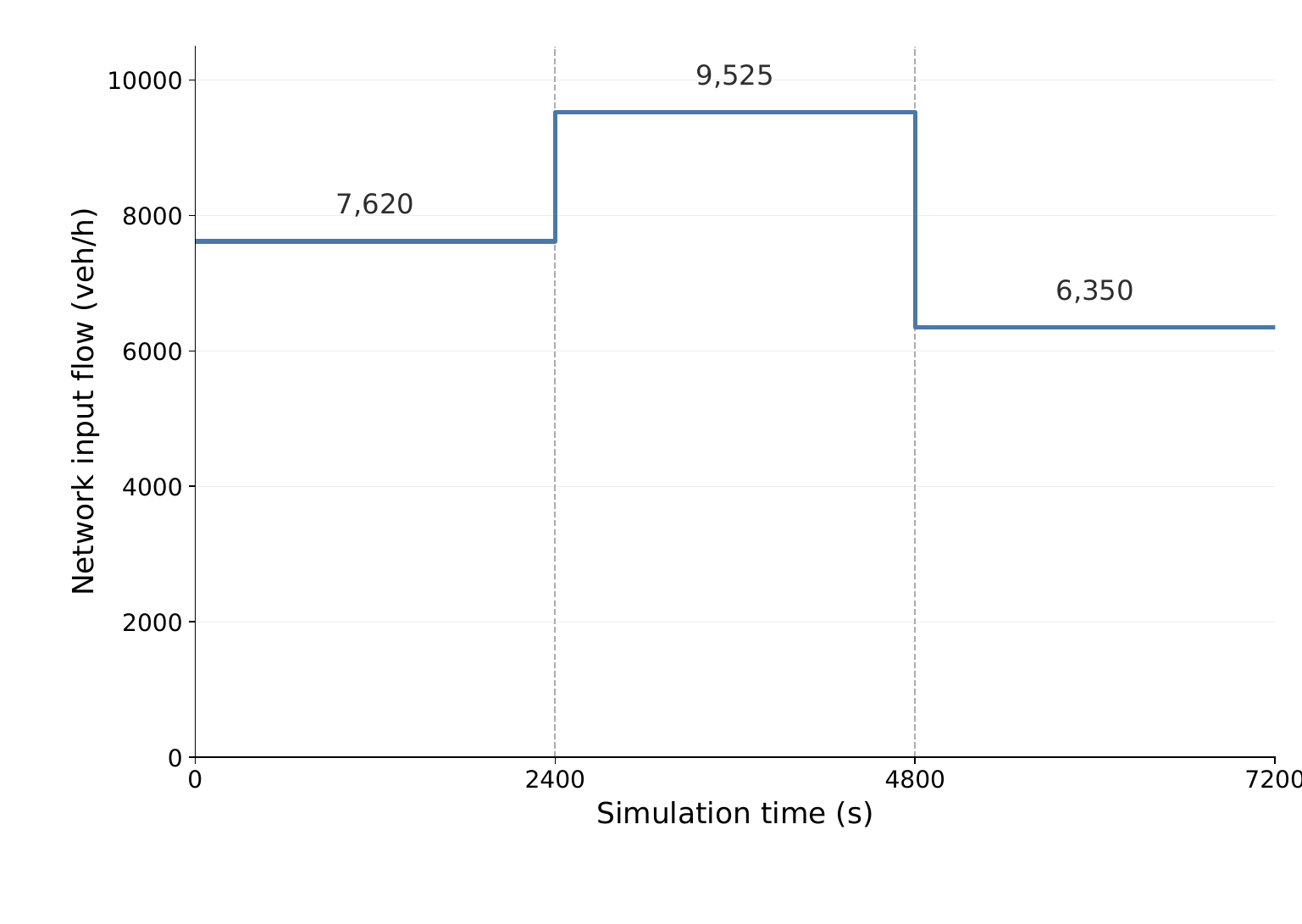} \\
{\small (a) Lianyungang test network.} & {\small (b) Time-varying network input demand.}
\end{tabular}
\caption{Microscopic simulation scenario. Panel (a) shows the 28-intersection Lianyungang test network, and Panel (b) shows the time-varying network input demand over the 7200~s control horizon.}
\label{fig:simulation_scenario}
\end{figure}

At each decision epoch, detector counts and flow conservation reconstruct each logical-link total, which is distributed among outgoing movements using turning proportions from a 300~s rolling window. The movement state is normalized by upstream-link length,
\(\widehat X_{ij}(t) = \frac{\widehat Q_{ij}(t)}{L_i},\)
with latent counterpart $X_{ij}(t)=Q_{ij}(t)/L_i$. This is the link-length special case of the fixed positive queue weighting in Appendix~\ref{app:weighted_stability}.

The uncertainty model is calibrated offline from three independent simulation runs, while three disjoint runs are used for baseline evaluation. At each intersection, residuals between the estimated and latent normalized movement states and turning ratios are reduced to eight representative joint scenarios. Stochastic Regret MP uses their empirical probabilities, whereas DR Regret MP allows the scenario probabilities to vary within a bootstrap-calibrated Wasserstein ambiguity set. The resulting calibration is fixed during each evaluation run.

Four controllers use the same link-length-weighted pressure described above: True and Point use the weighted counterparts of~\eqref{eq:true_state_mp} and~\eqref{eq:point_estimate_mp}, while Stochastic and DR use the corresponding weighted regret rules from~\eqref{eq:stochastic_controller} and~\eqref{eq:dr_controller}. Latent simulator states are used only for offline calibration and evaluation; Appendix~\ref{app:microscopic_estimator} gives estimator details.

The primary measures are mean delay (SUMO time loss), mean travel time, and the time-mean network total movement queue $\overline Q_{\mathrm{mov}}=|\mathcal I_{\mathrm{eval}}|^{-1}\sum_{t\in\mathcal I_{\mathrm{eval}}}\sum_m Q_m(t)$ for $\mathcal I_{\mathrm{eval}}=\{t:300\leq t\leq7200,\ t/15\in\mathbb Z\}$~s. Trip measures use completed vehicles; mean absolute error (MAE) summarizes estimation accuracy. We evaluate baseline performance, intervention timing, and sensing degradation.

\subsection{Baseline closed-loop performance}
\label{sec:simulation_baseline}

With complete unbiased detectors, Stochastic Regret MP and DR Regret MP reduce mean delay by 18.9\% and 14.9\%, respectively, relative to Point-estimate MP; both controllers also reduce travel time and the time-mean network total movement queue. True-state MP remains substantially better and therefore quantifies the performance still lost to imperfect state information. The computational cost is small relative to the 15~s update interval: a network-wide phase-selection update remains below 0.15~s for both proposed controllers.

\begin{figure}
\centering
\includegraphics[width=\linewidth]
{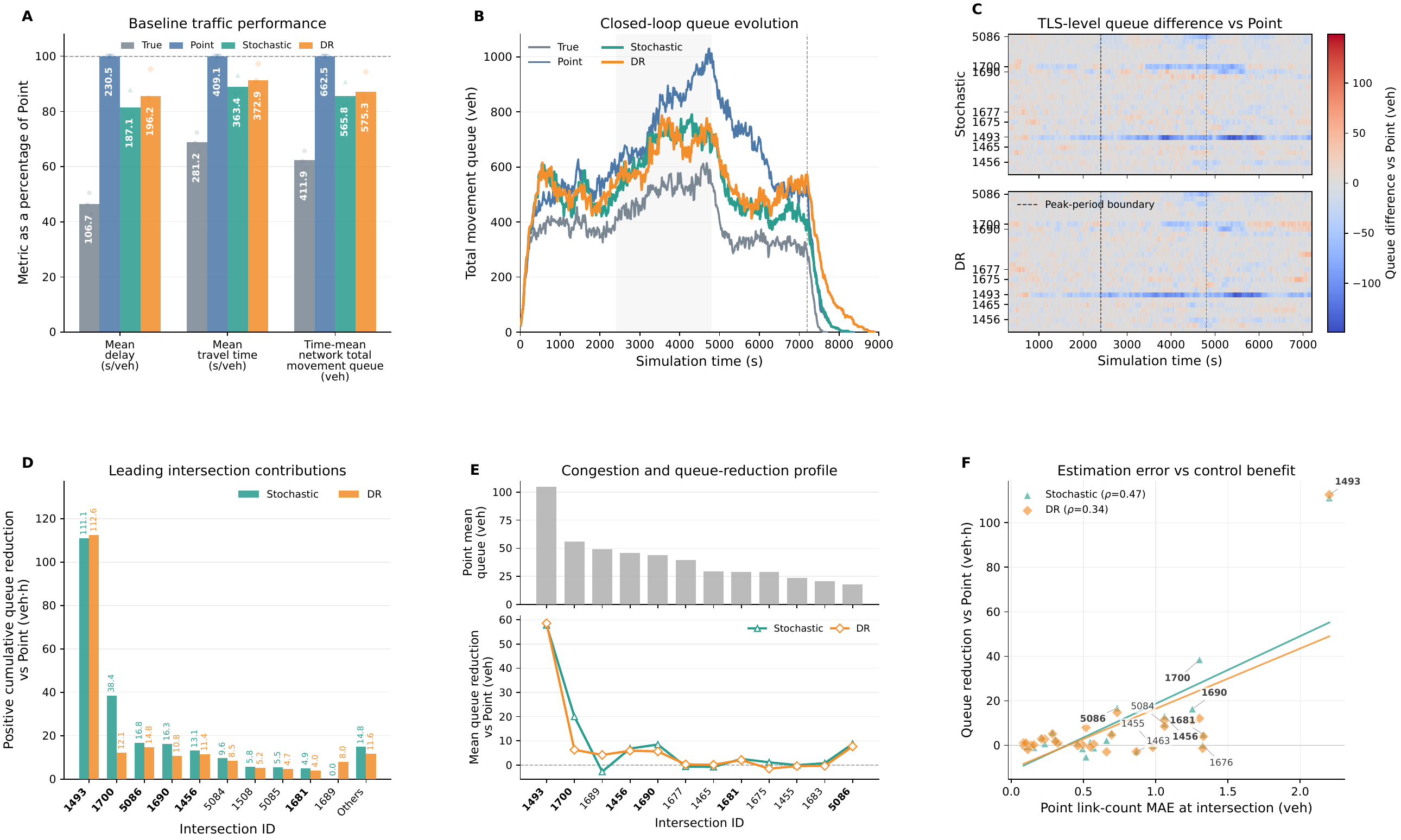}
\caption{Baseline closed-loop performance. Panel A compares mean delay, travel time, and time-mean network total movement queue over three evaluation runs as percentages of Point. Panel B shows the network total true movement queue for a representative evaluation run. Panel C shows intersection-level true queue differences relative to Point over time. Panel D reports the leading positive cumulative queue-reduction contributions. Panel E compares the Point mean queue with the mean queue reduction at the most congested intersections. Panel F relates Point estimation MAE to cumulative queue benefit. Intersection IDs appearing in the leading-benefit, high-queue, and high-error groups are shown in bold in Panels D--F.}
\label{fig:simulation_baseline_composite}
\end{figure}

Panel~B shows that the queue advantage does not arise from a single isolated decision. As congestion builds, repeated regret-aware phase choices alter movement discharge and downstream arrivals, so the subsequent traffic states under Stochastic Regret MP and DR Regret MP progressively diverge from the Point-estimate MP trajectory. The resulting benefit is therefore a closed-loop accumulation of locally improved phase rankings.

The spatial results in Panels~C--F sharpen this mechanism. For controller $\psi\in\{\mathrm{Sto},\mathrm{DR}\}$, let $Q_n^\psi(t)$ denote the total true movement queue at intersection $n$ and define its cumulative queue benefit relative to Point-estimate MP as $\mathrm{CQB}_n^\psi=\sum_{t\in\mathcal I_{\mathrm{eval}}}[Q_n^{\mathrm P}(t)-Q_n^\psi(t)]\Delta t$, with $\Delta t=15/3600$~h. Only five intersections account for 80\% of the positive Stochastic Regret MP benefit and six intersections account for 80\% of the positive DR Regret MP benefit. These high-impact locations also tend to have both larger queues and larger estimation errors. Hence the network gains are concentrated where an incorrect phase ranking is simultaneously more likely and more consequential, rather than being spread uniformly across intersections.

\subsection{Timing of regret-control intervention}
\label{sec:simulation_intervention}

The intervention experiment starts from the common Point-estimate MP trajectory and switches to Stochastic Regret MP or DR Regret MP at $t_s\in\{2400,3600,4800,6000\}$~s. All estimator, traffic, signal, and calibration states are carried through without reset, so the post-switch trajectories isolate the effect of changing the controller.

For controller $\psi\in\{\mathrm{Sto},\mathrm{DR}\}$, let $Q_{\mathrm{net}}^\psi(t)$ denote the network total true movement queue and define the post-intervention queue burden as $\mathrm{QB}^\psi(t_s)=3600^{-1}\int_{t_s}^{7200}Q_{\mathrm{net}}^\psi(t)\,dt$. The recoverable benefit is $\Delta\mathrm{QB}^\psi(t_s)=\mathrm{QB}^{\mathrm P}(t_s)-\mathrm{QB}^\psi(t_s)$, which measures how much cumulative congestion can still be avoided after the controller is activated.

\begin{figure}
\centering
\includegraphics[width=\linewidth]
{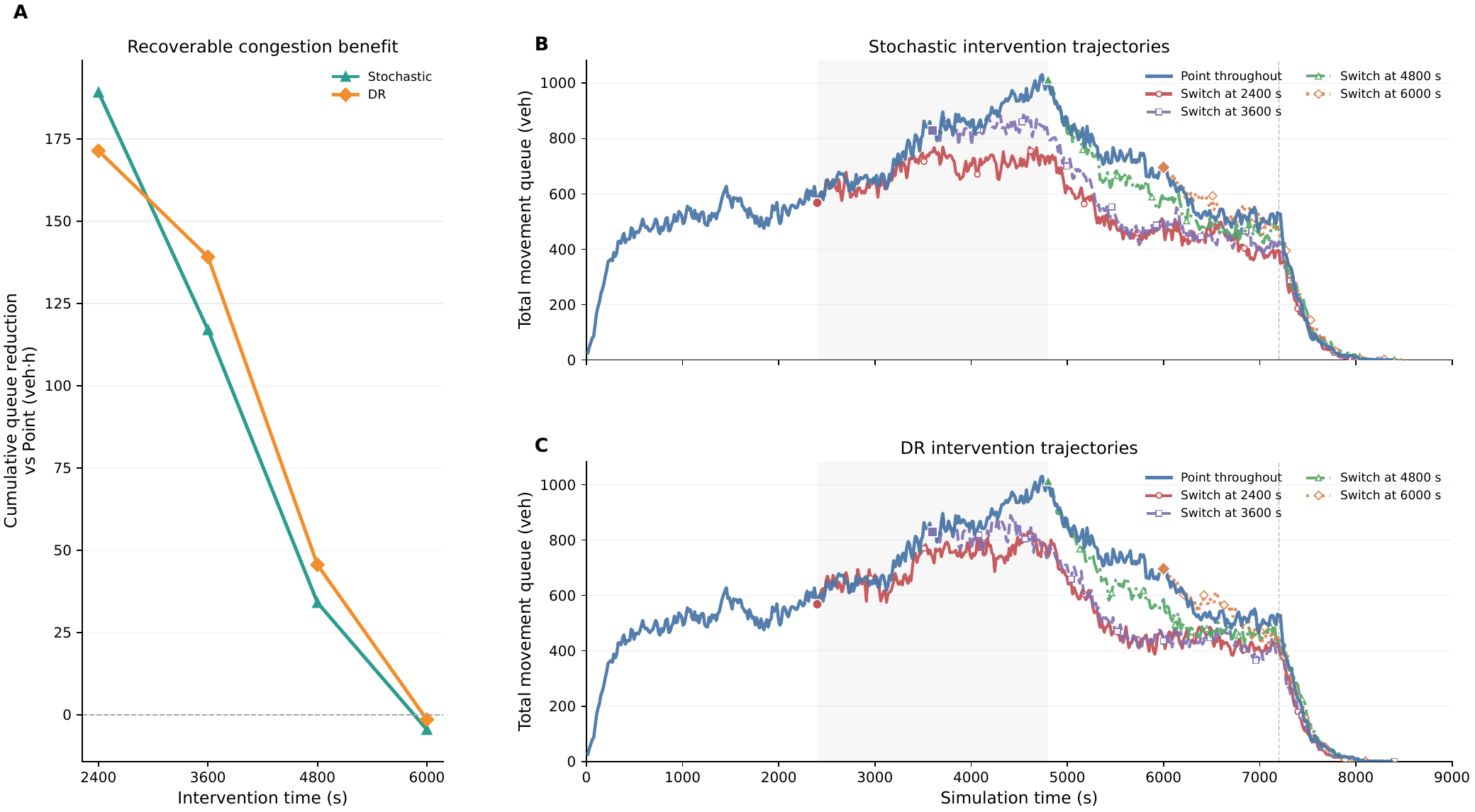}
\caption{Timing of regret-control intervention. Panel A shows the recoverable cumulative queue benefit relative to continuing Point. Panels B and C show the network total true movement-queue trajectories after switching from the common Point trajectory to Stochastic and DR, respectively, at different intervention times.}
\label{fig:simulation_intervention_composite}
\end{figure}

Figure~\ref{fig:simulation_intervention_composite} shows a strong timing effect. Activation at 2400~s recovers roughly 170--190~veh$\cdot$h of subsequent queue burden, whereas activation at 6000~s leaves too little time to offset congestion accumulated under Point-estimate MP. The trajectory panels explain why: earlier regret-aware decisions reshape downstream arrivals sooner, and these altered states affect many later decisions. The intervention benefit is therefore path dependent, with early local improvements propagating into much larger network-level gains.

\subsection{Performance under sensing degradation}
\label{sec:simulation_robustness}

The final experiment considers persistent detector undercount and detector failure. Each sensing condition is recalibrated from three runs and evaluated on a disjoint traffic realization across five spatial layouts, so the retained residual scenarios reflect the error structure produced by that sensing condition.

For detector undercount, a fixed subset of detectors is assigned persistent count bias throughout each run. We vary both the affected-detector fraction and the maximum undercount level, as shown in Fig.~\ref{fig:simulation_sensing_composite}. Stochastic Regret MP and DR Regret MP remain better than Point-estimate MP across the tested settings, indicating that scenario-based phase evaluation remains effective as systematic sensing bias increases state-estimation uncertainty.

\begin{figure}
\centering
\includegraphics[width=\linewidth]
{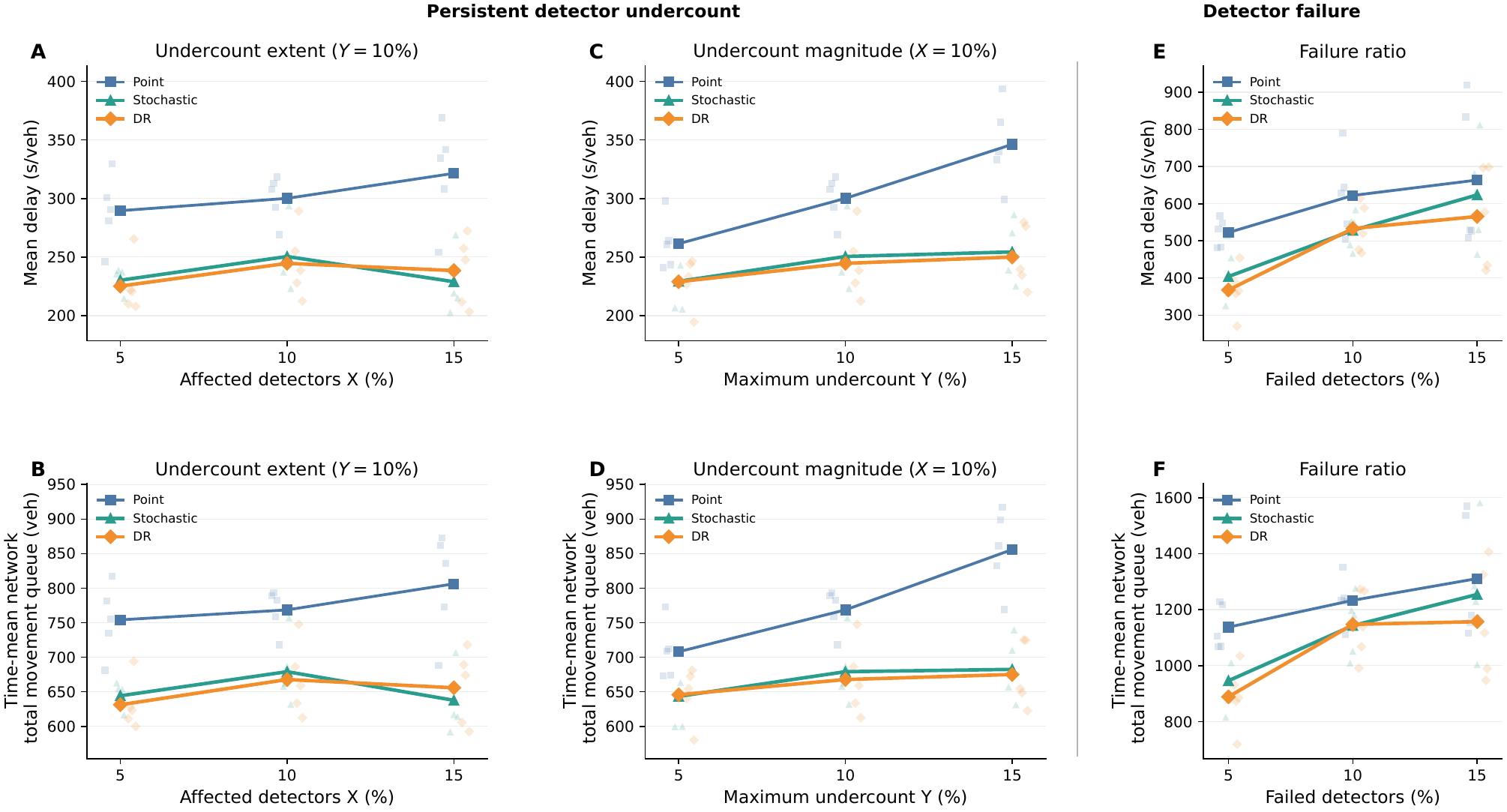}
\caption{Performance under sensing degradation. Panels A and B show mean delay and time-mean network total movement queue as the affected-detector fraction $p_{\mathrm{uc}}$ varies at $b_{\max}=10\%$; Panels C and D show the same measures as $b_{\max}$ varies at $p_{\mathrm{uc}}=10\%$. Panels E and F show mean delay and time-mean network total movement queue under increasing detector-failure rates. Faint points denote spatial detector realizations and lines denote realization means.}
\label{fig:simulation_sensing_composite}
\end{figure}

For detector failure, increasing fractions of eligible internal detectors are removed while network-entry detectors remain available; missing internal flows are reconstructed using the conservation-based estimator described in Appendix~\ref{app:microscopic_estimator}. The benefit of uncertainty-aware phase selection persists at the largest tested failure rate: with 15\% detector failure, Stochastic Regret MP and DR Regret MP reduce mean delay by 6.0\% and 14.8\%, respectively, relative to Point-estimate MP. The wider dispersion across spatial layouts also shows that performance becomes increasingly sensitive to where sensing information is lost.

Across the microscopic experiments, the same mechanism appears at three levels: the largest gains occur at congested, poorly estimated intersections; those local gains accumulate through closed-loop traffic evolution; and the benefit persists when sensing quality deteriorates. The SUMO results therefore connect the decision-level insight from the numerical experiments to network-level congestion reduction under a realistic detector-based implementation.

\section{Conclusion}
\label{sec:conclusion}

This paper developed a decision-oriented treatment of traffic-state uncertainty for max-pressure control. Phase-selection regret measures the pressure loss of a candidate phase relative to the state-specific pressure-maximizing phase, allowing heterogeneous uncertainty in queues, turning proportions, and service conditions to be evaluated on a common decision-loss scale. Stochastic Regret MP minimizes expected regret under the calibrated residual distribution, whereas DR Regret MP protects the same decision-relative objective against ambiguity in scenario probabilities. The analysis shows that expected-regret minimization preserves the expected-score decision under a fixed distribution, while regret removes state-dependent score components shared by all actions before uncertainty aggregation. Under distributional ambiguity, this property can make regret-based and absolute-pressure robustification select different phases. The Lyapunov analysis further establishes strong queue stability when the conditional net-service reserve exceeds the queue-linear uncertainty-loss coefficient, and the regret-margin result identifies when uncertainty perturbations cannot change the selected phase.

The experiments support these mechanisms at both decision and network levels. In the numerical experiments, Stochastic Regret MP and DR Regret MP reduce normalized true-state pressure regret by about 70\% and mean network queues by about 15\% relative to Point-estimate MP. The formulation and sensitivity tests show that the robustification target matters, but that differences among uncertainty treatments mainly affect low-margin phase decisions. In the 28-intersection SUMO network, Stochastic Regret MP and DR Regret MP reduce mean delay by 18.9\% and 14.9\%, respectively. The largest gains occur where heavy congestion coincides with larger estimation errors, and the intervention experiment shows that early local improvements accumulate through downstream traffic evolution. Both proposed controllers also retain benefits under detector undercount and detector failure, indicating that the decision-relative uncertainty representation remains useful as sensing quality deteriorates.

Several extensions remain important. First, future work can update the residual distribution and Wasserstein radius online as traffic regimes and sensing quality change, together with stability analysis for time-varying uncertainty models. Second, spatial dependence among neighboring estimation errors can be incorporated while preserving decentralized computation. Third, the framework can be extended to cyclic MP and joint phase-service allocation \citep{levin2020cyclic}. Finally, validation on larger heterogeneous networks, field data, and additional sensing modalities is needed to assess the transferability of residual calibration, ambiguity selection, and computational performance beyond the present simulation settings.

\section*{Declaration of generative AI and AI-assisted technologies in the writing process}

During the preparation of this work, the authors used ChatGPT in order to improve the language and assist with programming. After using this tool/service, the authors reviewed and edited the content as needed and take full responsibility for the content of the publication.

\bibliographystyle{plainnat}
\bibliography{dr_mp_references}

\newpage

\appendix

\section{Residual Calibration and Scenario Construction}
\label{app:residual_calibration}

This appendix gives the residual-calibration and scenario-reduction details summarized in Section~\ref{sec:uncertain_state}. Each calibration record pairs the traffic-state estimate with the corresponding latent traffic state. For calibration record $r=1,\ldots,N_{\mathrm{cal}}$, let $\bm\xi^r=\operatorname{col}(\bm\xi^{Q,r},\bm\xi^{R,r},\bm\xi^{C,r})$, where, for positive queue and service observations, $\xi_m^{Q,r}=\widehat Q_m^r/Q_m^r-1$, $\xi_m^{C,r}=\widehat c_m^r/c_m^r-1$, and $\bm\Xi^{R,r}=\widehat{\bm R}^r-\bm R^r$. When $Q_m^r=0$, the zero-queue state is taken as directly identifiable, so we set $\widehat Q_m^r=0$ and use $\xi_m^{Q,r}=0$ by convention. The empirical residual measure is $\widehat{\mathbb P}_{N_{\mathrm{cal}}}=\sum_{r=1}^{N_{\mathrm{cal}}}\omega_r\delta_{\bm\xi^r}$, where $\omega_r\geq0$ and $\sum_r\omega_r=1$.
Uniform $\omega_r=1/N_{\mathrm{cal}}$ gives the standard empirical distribution, while nonuniform weights can be used for regime balancing or time-of-day reweighting.

Routing realizations are constrained to the row-substochastic set
$$
\mathcal R
    =
    \left\{
        \bm R\succeq\bm0:
        \bm R\bm1\preceq\bm1,
        \ r_{ab}=0\ \text{for infeasible transitions}
    \right\}.
$$
This set is consistent with the convention that $r_{ab}$ denotes routing from movement $a$ to downstream movement $b$.

Because the residual components have different numerical scales, let $d_\xi$ denote the residual-vector dimension and let $\bm S_{\xi}=\operatorname{diag}(s_1,\ldots,s_{d_\xi})$ with $s_\ell>0$. Define the normalized weighted $L_1$ metric $d(\bm\xi,\bm\zeta)=d_\xi^{-1}\|\bm S_\xi^{-1}(\bm\xi-\bm\zeta)\|_1$.
The scales can be chosen from robust empirical spreads or physically meaningful error ranges.

For online control, the empirical distribution is compressed to $K_{\mathrm{sc}}$ representative joint residuals
\(\bar\Xi_{K_{\mathrm{sc}}} = \{\bar{\bm\xi}^{1},\ldots,\bar{\bm\xi}^{K_{\mathrm{sc}}}\} \subseteq \{\bm\xi^1,\ldots,\bm\xi^{N_{\mathrm{cal}}}\}.\)
Each calibration record is assigned to its closest representative,
\(\iota(r) \in \arg\min_{k=1,\ldots,K_{\mathrm{sc}}} d(\bm\xi^r,\bar{\bm\xi}^{k}),\)
and the representative probabilities are
\(\widehat\pi_k = \sum_{r:\iota(r)=k}\omega_r.\)
The resulting reduced empirical measure is $\widehat{\mathbb P}_{K_{\mathrm{sc}}}=\sum_{k=1}^{K_{\mathrm{sc}}}\widehat\pi_k\delta_{\bar{\bm\xi}^{k}}$.
A forward scenario-reduction procedure can select the representative set by successively reducing weighted quantization distortion \citep{dupacova2003scenario,heitsch2003scenario}. Because complete residual vectors are retained, dependence among queue, routing, and service errors is preserved.

For each intersection, the reduced residuals are projected to the controlled movements and the immediate downstream movements required by the local pressure calculation. Duplicate local projections are merged and their probabilities added, yielding the local nominal distribution $\widehat{\mathbb P}_n$ in~\eqref{eq:local_nominal_distribution}. The same local metric induces the pairwise scenario distances $d_{n,ij}$ used by the Wasserstein ambiguity set in Section~\ref{sec:dr_regret}.

\section{Proofs and Derivations for Regret-Based Control}
\label{app:regret_control_proofs}

\subsection{Score-gap equivalence and robustification structure}
\label{app:regret_structural_proofs}

\begin{proof}[Proof of Proposition~\ref{prop:fixed_distribution_equivalence}]
By definition, $L(a,\bm\xi)=M(\bm\xi)-S(a,\bm\xi)$. Therefore, for every $a\in\mathcal A$,
$$
\mathbb E_{\mathbb P}[L(a,\bm\xi)]
=
\mathbb E_{\mathbb P}[M(\bm\xi)]
-
\mathbb E_{\mathbb P}[S(a,\bm\xi)].
$$
The first term is independent of $a$, so minimizing expected regret and maximizing expected score have exactly the same optimizer set. This proves~\eqref{eq:fixed_distribution_equivalence}.
\end{proof}

\begin{proof}[Proof of Proposition~\ref{prop:common_shift_invariance}]
Let $\widetilde M(\bm\xi)=\max_{u\in\mathcal A}\widetilde S(u,\bm\xi)$. Since the added term is independent of the action,
$$
\widetilde M(\bm\xi)
=
\max_u\{S(u,\bm\xi)+c(\bm\xi)\}
=
M(\bm\xi)+c(\bm\xi).
$$
Hence
$$
\widetilde L(a,\bm\xi)
=
\widetilde M(\bm\xi)-\widetilde S(a,\bm\xi)
=
M(\bm\xi)-S(a,\bm\xi)
=
L(a,\bm\xi).
$$
Any fixed-distribution expectation, supremum of expected regret over a fixed ambiguity set, or support-wise maximum applied to the same regret function is therefore unchanged.
\end{proof}

\begin{proof}[Proof of Proposition~\ref{prop:robustification_distinction}]
If $\mathbb E_{\mathbb P}[M(\bm\xi)]=\bar M$ for every $\mathbb P\in\mathfrak P$, then for each action $a$,
$$
J_{\mathrm{reg}}(a;\mathfrak P)
=
\sup_{\mathbb P\in\mathfrak P}\{\bar M-\mathbb E_{\mathbb P}[S(a,\bm\xi)]\}
=
\bar M-\inf_{\mathbb P\in\mathfrak P}\mathbb E_{\mathbb P}[S(a,\bm\xi)]
=
\bar M-J_{\mathrm{abs}}(a;\mathfrak P),
$$
which proves~\eqref{eq:robust_equivalence_condition} and the equality of the optimizer sets after reversing the sign.

To show that decision equivalence need not hold in general, consider two actions $a_1,a_2$, two states $\xi_1,\xi_2$, and the score table
$$
\begin{array}{c|cc}
 & \xi_1 & \xi_2\\ \hline
a_1 & 2 & 98\\
a_2 & 0 & 102
\end{array}.
$$
Then $M(\xi_1)=2$, $M(\xi_2)=102$, and the corresponding regrets are
$$
\begin{array}{c|cc}
 & \xi_1 & \xi_2\\ \hline
a_1 & 0 & 4\\
a_2 & 2 & 0
\end{array}.
$$
Let the nominal distribution be $(1/2,1/2)$, let the distance between the two states be one, and let $\mathfrak P(\varepsilon)$ be the resulting discrete 1-Wasserstein ball with $0\leq\varepsilon\leq1/2$. Writing $p_1=\mathbb P(\xi_1)$, the ambiguity set is $p_1\in[1/2-\varepsilon,1/2+\varepsilon]$. Therefore
$$
J_{\mathrm{reg}}(a_1;\mathfrak P(\varepsilon))=2+4\varepsilon,
\qquad
J_{\mathrm{reg}}(a_2;\mathfrak P(\varepsilon))=1+2\varepsilon,
$$
so robust regret selects $a_2$. For the absolute scores,
$$
J_{\mathrm{abs}}(a_1;\mathfrak P(\varepsilon))=50-96\varepsilon,
\qquad
J_{\mathrm{abs}}(a_2;\mathfrak P(\varepsilon))=51-102\varepsilon.
$$
Whenever $\varepsilon>1/6$, $J_{\mathrm{abs}}(a_1;\mathfrak P(\varepsilon))>J_{\mathrm{abs}}(a_2;\mathfrak P(\varepsilon))$, so robust absolute-score maximization selects $a_1$. Thus the two robustification criteria can select different actions under the same finite Wasserstein ambiguity set.
\end{proof}

\begin{proof}[Proof of Lemma~\ref{lem:dr_regret_dual}]
Associate the nonnegative multiplier $\nu$ with the Wasserstein transport-budget constraint in~\eqref{eq:dr_regret_transport} and a free multiplier $\theta_i$ with each row-sum constraint. The dual of the finite transport linear program is
$$
\begin{array}{ll}
\displaystyle \min_{\nu\geq0,\bm\theta} & \displaystyle \nu\varepsilon_n+\sum_i\widehat\pi_{n,i}\theta_i\\[2pt]
\text{s.t.} & \theta_i\geq L_{n,j}(\bm s_n,t)-\nu d_{n,ij},\quad \forall i,j.
\end{array}
$$
For fixed $\nu$, the smallest feasible row multiplier is $\theta_i=\max_j\{L_{n,j}(\bm s_n,t)-\nu d_{n,ij}\}=g_{n,i}(\nu;\bm s_n,t)$. Substitution gives~\eqref{eq:dr_regret_dual}. Strong duality applies because~\eqref{eq:dr_regret_transport} is a feasible and bounded finite linear program.
\end{proof}

\begin{proof}[Proof of Proposition~\ref{prop:regret_ordering}]
Because $\widehat{\bm\pi}_n\in\mathcal P_n(\varepsilon_n)$, the supremum over the ambiguity set cannot be smaller than the expected regret evaluated at the calibrated distribution. Moreover, every $\bm p_n\in\mathcal P_n(\varepsilon_n)\subseteq\Delta_{K_n}$ produces a convex combination of the scenario regrets and therefore cannot exceed their maximum. Hence $G_n^{\mathrm{Sto}}\leq G_n^{\mathrm{DR}}\leq G_n^{\mathrm{Set}}$, proving~\eqref{eq:regret_ordering}.
\end{proof}

\begin{proof}[Proof of Proposition~\ref{prop:dr_regret_limits}]
When $\varepsilon_n=0$, the discrete Wasserstein ball contains only $\widehat{\bm\pi}_n$, so the DR score equals the stochastic score. For any $j\in\mathcal K_n^{\max}(\bm s_n,t)$, concentrating all nominal probability mass on scenario $j$ has transport cost $\sum_i\widehat\pi_{n,i}d_{n,ij}$. Therefore, if $\varepsilon_n\geq\varepsilon_{n,\bm s}^{\mathrm{Set}}(t)$, the ambiguity set contains a distribution concentrated on a maximum-regret scenario. Its expected regret is $G_n^{\mathrm{Set}}(\bm s_n,t)$; Proposition~\ref{prop:regret_ordering} then forces equality with the DR score.
\end{proof}

\section{Proofs and Extensions for Network Stability}
\label{app:network_stability_proofs}

This appendix provides the technical results underlying Section~\ref{sec:theory}. The derivations follow the order of the main analysis: the common drift framework, conditional pressure approximation, the queue-linear pressure-loss bound and strong stability, and the comparison with Point-estimate MP. The fixed-positive-weight extension is given at the end.

\subsection{Drift framework and perfect-information benchmark}
\label{app:drift_proof}
\label{app:perfect_info_proofs}

\begin{proof}[Proof of Lemma~\ref{lem:quadratic_drift}]
Define $\bm Y(t)=\bm A(t)-[\bm I-\bm R^\top(t)]\bm D(t)$. Then $\bm Q(t+1)=\bm Q(t)+\bm Y(t)$ and
\begin{equation}
    V(\bm Q(t+1))-V(\bm Q(t))
    =
    \bm Q^\top(t)\bm Y(t)
    +
    \frac12\|\bm Y(t)\|_2^2.
    \label{eq:app_exact_drift}
\end{equation}
Because each routing row is nonnegative with sum at most one, the finite network admits $K_R=\sup_t\|\bm I-\bm R^\top(t)\|_2<\infty$. Moreover, $0\leq D_m(t)\leq c_m^{\max}$ and $\|\bm D(t)\|_2^2\leq\sum_m(c_m^{\max})^2\triangleq\bar C_2$. Using $\|\bm a-\bm b\|_2^2\leq2\|\bm a\|_2^2+2\|\bm b\|_2^2$ and Assumption~\ref{ass:bounded_primitives},
$$
\frac12
\mathbb E[\|\bm Y(t)\|_2^2\mid\mathcal H_t]
\leq
\bar A_2+K_R^2\bar C_2
\triangleq B_{\mathrm q}.
$$
Since $\bm Q(t)$ is $\mathcal H_t$-measurable,
$\mathbb E[\bm Q^\top(t)\bm A(t)\mid\mathcal H_t]=\bm Q^\top(t)\bm\lambda$.

To relate realized discharge to potential pressure, write $\bm U=\bm C\bm s$ and $\bm e^{\mathrm{us}}=\bm U-\bm D\succeq\bm0$. Then
$$
\bm Q^\top(\bm I-\bm R^\top)\bm D
=
H(\bm s,t)-\bm Q^\top(\bm I-\bm R^\top)\bm e^{\mathrm{us}}
\geq
H(\bm s,t)-\bm Q^\top\bm e^{\mathrm{us}}.
$$
Whenever $e_m^{\mathrm{us}}>0$, $Q_m<U_m\leq c_m^{\max}$ and $e_m^{\mathrm{us}}\leq c_m^{\max}$, so $Q_me_m^{\mathrm{us}}\leq(c_m^{\max})^2$. Hence, with $B_{\mathrm u}=\sum_m(c_m^{\max})^2$,
$$
\bm Q^\top(\bm I-\bm R^\top)\bm D
\geq
H(\bm s,t)-B_{\mathrm u}.
$$
Substitution into~\eqref{eq:app_exact_drift} proves~\eqref{eq:initial_drift_bound} with the valid choice $B_0=B_{\mathrm q}+B_{\mathrm u}$.
\end{proof}

\begin{proof}[Proof of Corollary~\ref{cor:perfect_information_benchmark}]
When uncertainty collapses to a single latent state, phase regret is
$$
\max_{\bm u_n\in\mathcal S_n}H_n(\bm u_n,t)-H_n(\bm s_n,t).
$$
The first term is independent of the candidate phase, so minimizing regret is equivalent to maximizing the true-state phase pressure. This establishes the phase-selection statement for Stochastic, DR, and Set Regret MP.

For stability, let $\mu\in\mathfrak U_{\mathrm{phy}}$ satisfy $\bar{\bm Z}^{\mu}\succeq\bm\lambda+\epsilon_0\bm1$, and let $\bm s^{\mathrm T}(t)$ denote the resulting true-state MP phase. For every realization of the current network-service state,
\begin{equation}
H(\bm s^{\mathrm T}(t),t)\geq H(\bm s^{\mu}(t),t).
\label{eq:instantaneous_true_mp_domination}
\end{equation}
Let $\mathcal F_t^{-}\subseteq\mathcal H_t$ denote the queue and history available before the current i.i.d. pair $(\bm R(t),\bm C(t))$ is realized. Because this pair is exogenous and i.i.d., $\bm Q(t)$ is $\mathcal F_t^{-}$-measurable and $\mathbb E[\bm Z^\mu(t)\mid\mathcal F_t^{-}]=\bar{\bm Z}^{\mu}$. Therefore,
$$
\mathbb E\!\left[
H(\bm s^{\mathrm T}(t),t)-\bm Q^\top(t)\bm\lambda
\mid\mathcal F_t^{-}
\right]
\geq
\bm Q^\top(t)(\bar{\bm Z}^{\mu}-\bm\lambda)
\geq
\epsilon_0\|\bm Q(t)\|_1.
$$
Taking the conditional expectation of~\eqref{eq:initial_drift_bound} with respect to $\mathcal F_t^{-}$ and then expectations gives
$\mathbb E[\Delta V(t)]\leq B_0-\epsilon_0\mathbb E\|\bm Q(t)\|_1$.
Summing over time, using $V(\bm Q(T))\geq0$, dividing by $T\epsilon_0$, and taking the limit superior yields~\eqref{eq:perfect_information_queue_bound}.
\end{proof}

\subsection{Conditional pressure approximation}
\label{app:conditional_regret_bounds}

\begin{proof}[Proof of Lemma~\ref{lem:conditional_regret_pressure}]
Taking the conditional expectation of~\eqref{eq:continuous_regret} gives
$$
\bar L_n(\bm s_n,t)
=
\mathbb E\!\left[
\max_{\bm u_n\in\mathcal S_n}H_n(\bm u_n,\bm\xi_n(t),t)
\mid\mathcal G_t
\right]
-
\bar H_n(\bm s_n,t).
$$
The first term is independent of the candidate phase. Minimizing $\bar L_n$ is therefore equivalent to maximizing $\bar H_n$, and subtracting the identity for a conditional-pressure maximizer $\bm s_n^\star$ from that for an arbitrary phase gives the stated gap relation.
\end{proof}

The following technical lemma bounds the difference between the aggregate regret used by each controller and the true conditional regret.
\begin{lemma}[Finite-scenario regret approximation]
\label{lem:regret_score_bounds}
Under Assumption~\ref{ass:coverage}, for every candidate phase $\bm s_n$,
\begin{equation}
    \left|G_n^{\mathrm{Sto}}-\bar L_n\right|
    \leq
    (\varepsilon_n^{\mathrm{cov}}+\tau_n)\Lambda_n(t),
    \label{eq:sto_true_regret_bound}
\end{equation}
\begin{equation}
    -\tau_n\Lambda_n(t)
    \leq G_n^{\mathrm{DR}}-\bar L_n
    \leq
    (\varepsilon_n+\varepsilon_n^{\mathrm{cov}}+\tau_n)\Lambda_n(t),
    \label{eq:dr_true_regret_bound}
\end{equation}
and
\begin{equation}
    -\tau_n\Lambda_n(t)
    \leq G_n^{\mathrm{Set}}-\bar L_n
    \leq
    (d_n^{\max}+\tau_n)\Lambda_n(t),
    \label{eq:set_true_regret_bound}
\end{equation}
where the common arguments $(\bm s_n,t)$ of the regret scores are suppressed.
\end{lemma}

\begin{proof}[Proof of Lemma~\ref{lem:regret_score_bounds}]
Let $\bm\xi_n^{\mathrm{rep}}(t)=\bar{\bm\xi}_n^{\iota_n(\bm\xi_n(t))}$ and define the conditional representative regret
$$
\bar L_n^{\mathrm{rep}}(\bm s_n,t)
=
\sum_{k=1}^{K_n}\pi_{n,k}(t)L_{n,k}(\bm s_n,t).
$$
By~\eqref{eq:regret_lipschitz_modulus} and~\eqref{eq:quantization_distortion_bound},
\begin{equation}
\left|
\bar L_n(\bm s_n,t)-\bar L_n^{\mathrm{rep}}(\bm s_n,t)
\right|
\leq
\tau_n\Lambda_n(t).
\label{eq:continuous_quantized_regret_bound}
\end{equation}
On the retained support, the scenario-regret function is $\Lambda_n(t)$-Lipschitz with respect to $d_n$. Kantorovich--Rubinstein duality therefore gives, for probability vectors $\bm p$ and $\bm q$ on this support,
\begin{equation}
\left|
\sum_kp_kL_{n,k}-\sum_kq_kL_{n,k}
\right|
\leq
\Lambda_n(t)W_n(\bm p,\bm q).
\label{eq:kr_regret}
\end{equation}
For Stochastic Regret MP, applying~\eqref{eq:kr_regret} to $\widehat{\bm\pi}_n$ and $\bm\pi_n(t)$ and then using~\eqref{eq:coverage_condition} and~\eqref{eq:continuous_quantized_regret_bound} proves~\eqref{eq:sto_true_regret_bound}.

For DR Regret MP, $\varepsilon_n\geq\varepsilon_n^{\mathrm{cov}}$ makes $\bm\pi_n(t)$ feasible in $\mathcal P_n(\varepsilon_n)$, so $G_n^{\mathrm{DR}}\geq\bar L_n^{\mathrm{rep}}$. For any $\bm p_n\in\mathcal P_n(\varepsilon_n)$,
$$
W_n(\bm p_n,\bm\pi_n(t))
\leq
W_n(\bm p_n,\widehat{\bm\pi}_n)
+
W_n(\widehat{\bm\pi}_n,\bm\pi_n(t))
\leq
\varepsilon_n+\varepsilon_n^{\mathrm{cov}}.
$$
Equation~\eqref{eq:kr_regret} then gives
$$
0\leq G_n^{\mathrm{DR}}-\bar L_n^{\mathrm{rep}}
\leq
(\varepsilon_n+\varepsilon_n^{\mathrm{cov}})\Lambda_n(t),
$$
which together with~\eqref{eq:continuous_quantized_regret_bound} proves~\eqref{eq:dr_true_regret_bound}.

For Set Regret, the difference between the maximum representative regret and any convex combination of representative regrets is bounded by the range of the regrets. Since the retained support has diameter $d_n^{\max}$,
$$
0\leq
G_n^{\mathrm{Set}}-\bar L_n^{\mathrm{rep}}
\leq
\max_kL_{n,k}-\min_kL_{n,k}
\leq
d_n^{\max}\Lambda_n(t).
$$
Combining this with~\eqref{eq:continuous_quantized_regret_bound} proves~\eqref{eq:set_true_regret_bound}.
\end{proof}

\begin{proof}[Proof of Proposition~\ref{prop:local_approx_pressure}]
Let $\bm s_n^\star$ maximize $\bar H_n$. For Stochastic Regret MP, Lemma~\ref{lem:regret_score_bounds} and optimality of $\bm s_n^{\mathrm{Sto}}$ give
$$
\begin{array}{rcl}
\bar L_n(\bm s_n^{\mathrm{Sto}})
&\leq&G_n^{\mathrm{Sto}}(\bm s_n^{\mathrm{Sto}})+(\varepsilon_n^{\mathrm{cov}}+\tau_n)\Lambda_n\\
&\leq&G_n^{\mathrm{Sto}}(\bm s_n^\star)+(\varepsilon_n^{\mathrm{cov}}+\tau_n)\Lambda_n\\
&\leq&\bar L_n(\bm s_n^\star)+2(\varepsilon_n^{\mathrm{cov}}+\tau_n)\Lambda_n.
\end{array}
$$
For DR Regret MP,
$$
\begin{array}{rcl}
\bar L_n(\bm s_n^{\mathrm{DR}})
&\leq&\bar L_n^{\mathrm{rep}}(\bm s_n^{\mathrm{DR}})+\tau_n\Lambda_n\\
&\leq&G_n^{\mathrm{DR}}(\bm s_n^{\mathrm{DR}})+\tau_n\Lambda_n\\
&\leq&G_n^{\mathrm{DR}}(\bm s_n^\star)+\tau_n\Lambda_n\\
&\leq&\bar L_n^{\mathrm{rep}}(\bm s_n^\star)+(\varepsilon_n+\varepsilon_n^{\mathrm{cov}}+\tau_n)\Lambda_n\\
&\leq&\bar L_n(\bm s_n^\star)+(\varepsilon_n+\varepsilon_n^{\mathrm{cov}}+2\tau_n)\Lambda_n.
\end{array}
$$
For Set Regret, the same argument gives
$\bar L_n(\bm s_n^{\mathrm{Set}})-\bar L_n(\bm s_n^\star)\leq(d_n^{\max}+2\tau_n)\Lambda_n$.
Lemma~\ref{lem:conditional_regret_pressure} converts these regret gaps into~\eqref{eq:local_approx_pressure}.
\end{proof}

\subsection{Queue-linear pressure loss and strong stability}
\label{app:queue_linear_proof}
\label{app:main_stability_proof}

We first verify the affine sensitivity form for the residual reconstructions used in the paper. For the multiplicative reconstruction in~\eqref{eq:scenario_qc_reconstruction}, compactness and the positivity bounds in Section~\ref{sec:uncertain_state} give
$$
\left|\frac{\widehat Q_m}{1+\xi_m^Q}-\frac{\widehat Q_m}{1+\zeta_m^Q}\right|
\leq
\frac{\widehat Q_m}{\underline\chi_Q^2}|\xi_m^Q-\zeta_m^Q|,
\qquad
\left|\frac{\widehat c_m}{1+\xi_m^C}-\frac{\widehat c_m}{1+\zeta_m^C}\right|
\leq
\frac{\widehat c_m}{\underline\chi_C^2}|\xi_m^C-\zeta_m^C|.
$$
Bounded saturation service and compact residual support give a uniform bound on reconstructed service. Euclidean projection onto the closed convex feasible routing set is nonexpansive, and finite-dimensional norm equivalence with the scaling in $d_n$ therefore yields a finite $h_n$ such that
$$
|H_n(\bm s_n,\bm\xi,t)-H_n(\bm s_n,\bm\zeta,t)|
\leq
h_n\|\widehat{\bm Q}_{\mathcal L_n}(t)\|_1d_n(\bm\xi,\bm\zeta).
$$
Thus~\eqref{eq:affine_pressure_sensitivity} holds with $\beta_n=0$. The multiplicative queue relation also gives~\eqref{eq:affine_estimator_growth} with $b_{Q,n}=0$ and finite $\eta_{Q,n}$.

For bounded additive and clipped reconstruction, nonexpansiveness of projection onto the nonnegative orthant gives $|[\widehat Q_m-e_m]^+-[\widehat Q_m-e'_m]^+|\leq|e_m-e'_m|$. Bounded additive residuals contribute a finite queue-independent term to pressure sensitivity, while routing perturbations multiply a queue magnitude bounded by $\|\widehat{\bm Q}\|_1$ plus a finite residual bound. Hence finite $\beta_n$ and $h_n$ exist in~\eqref{eq:affine_pressure_sensitivity}. Likewise, $\widehat{\bm Q}=\bm Q+\bm e$ with bounded $\bm e$ gives~\eqref{eq:affine_estimator_growth} with finite $b_{Q,n}$ and $\eta_{Q,n}=1$; clipping can only reduce the reconstructed nonnegative magnitude.

\begin{proof}[Proof of Proposition~\ref{prop:queue_linear_pressure_loss}]
For any candidate phase, both the scenario-best pressure and the candidate pressure vary by at most $\Upsilon_n(t)d_n(\bm\xi,\bm\zeta)$. Hence
$\Lambda_n(t)\leq2\Upsilon_n(t)$.
Assumption~\ref{ass:affine_reconstruction} then gives
$$
\Lambda_n(t)
\leq
2\beta_n+2h_n b_{Q,n}+2h_n\eta_{Q,n}\|\bm Q_{\mathcal L_n}(t)\|_1
=
\gamma_n+\ell_n\|\bm Q_{\mathcal L_n}(t)\|_1,
$$
which proves~\eqref{eq:queue_linear_regret}. Summing over intersections,
$$
\sum_{n\in\mathcal N}a_n^\psi\Lambda_n(t)
\leq
\sum_na_n^\psi\gamma_n
+
\sum_mQ_m(t)\sum_{n:m\in\mathcal L_n}a_n^\psi\ell_n
\leq
B_\xi^\psi+\delta^\psi\|\bm Q(t)\|_1,
$$
which is~\eqref{eq:global_pressure_loss_bound}. Since the left-hand side is $\mathcal G_t$-measurable, taking the conditional expectation of the right-hand side given $\mathcal G_t$ yields~\eqref{eq:conditional_global_pressure_loss_bound}.
\end{proof}

\begin{proof}[Proof of Theorem~\ref{thm:regret_mp_stability}]
Because $\bm s^\psi(t)$ is $\mathcal G_t$-measurable, summing Proposition~\ref{prop:local_approx_pressure} over intersections and using~\eqref{eq:conditional_global_pressure_loss_bound} gives
$$
\mathbb E[H(\bm s^\psi(t),t)\mid\mathcal G_t]
\geq
\sum_{n\in\mathcal N}\max_{\bm s_n\in\mathcal S_n}\bar H_n(\bm s_n,t)
-B_\xi^\psi
-\delta^\psi\mathbb E[\|\bm Q(t)\|_1\mid\mathcal G_t].
$$
The comparison phase $\widetilde{\bm s}(t)$ is randomized using $\mathcal G_t$ only, so its conditional expected pressure is a convex combination of the local conditional pressures. Therefore,
$$
\sum_n\max_{\bm s_n}\bar H_n(\bm s_n,t)
\geq
\mathbb E[H(\widetilde{\bm s}(t),t)\mid\mathcal G_t].
$$
By~\eqref{eq:conditional_service_reserve},
$$
\mathbb E\!\left[
H(\widetilde{\bm s}(t),t)-\bm Q^\top(t)\bm\lambda
\mid\mathcal G_t^Q
\right]
\geq
\epsilon_0\|\bm Q(t)\|_1.
$$
Taking the conditional expectation given $\mathcal G_t$ and applying the tower property yields
$$
\mathbb E\!\left[
H(\widetilde{\bm s}(t),t)-\bm Q^\top(t)\bm\lambda
\mid\mathcal G_t
\right]
\geq
\epsilon_0\mathbb E[\|\bm Q(t)\|_1\mid\mathcal G_t].
$$
Combining the preceding inequalities gives
$$
\mathbb E[\bm Q^\top(t)\bm\lambda-H(\bm s^\psi(t),t)\mid\mathcal G_t]
\leq
B_\xi^\psi-(\epsilon_0-\delta^\psi)\mathbb E[\|\bm Q(t)\|_1\mid\mathcal G_t].
$$
Taking the conditional expectation of~\eqref{eq:initial_drift_bound} given $\mathcal G_t$, applying this bound, and then taking expectations gives
$$
\mathbb E[\Delta V(t)]
\leq
B_0+B_\xi^\psi-(\epsilon_0-\delta^\psi)\mathbb E\|\bm Q(t)\|_1.
$$
Since $\epsilon_0>\delta^\psi$, summing from $t=0$ to $T-1$, using $V(\bm Q(T))\geq0$ and $\mathbb E V(\bm Q(0))<\infty$, dividing by $T(\epsilon_0-\delta^\psi)$, and taking the limit superior proves~\eqref{eq:average_queue_bound}.
\end{proof}

\subsection{Point-estimate comparison and fixed positive queue weighting}
\label{app:point_comparison}

For Point-estimate MP, let $\widehat H_n(\bm s_n,t)=H_n(\bm s_n,\bm0,t)$. By the definition of $\Upsilon_n(t)$,
$$
|H_n(\bm s_n,\bm\xi,t)-\widehat H_n(\bm s_n,t)|
\leq
\Upsilon_n(t)d_n(\bm\xi,\bm0).
$$
Taking conditional expectations gives
$$
|\bar H_n(\bm s_n,t)-\widehat H_n(\bm s_n,t)|
\leq
\Upsilon_n(t)\rho_n^{\mathrm P}(t).
$$
Let $\bm s_n^\star$ maximize $\bar H_n$ and $\bm s_n^{\mathrm P}$ maximize $\widehat H_n$. Point optimality implies
$$
\bar H_n(\bm s_n^\star)-\bar H_n(\bm s_n^{\mathrm P})
\leq
2\Upsilon_n(t)\rho_n^{\mathrm P}(t).
$$
Because $2\Upsilon_n(t)\leq\gamma_n+\ell_n\|\bm Q_{\mathcal L_n}(t)\|_1$ almost surely and the left-hand side of this sensitivity inequality is $\mathcal G_t$-measurable,
$$
\bar H_n(\bm s_n^\star)-\bar H_n(\bm s_n^{\mathrm P})
\leq
\bar\rho_n^{\mathrm P}
\left[
\gamma_n+
\ell_n\mathbb E[\|\bm Q_{\mathcal L_n}(t)\|_1\mid\mathcal G_t]
\right].
$$
Summing over intersections gives the network-wide Point-estimate pressure-loss bound stated in Section~\ref{sec:point_stability_comparison}. Substituting this bound into the same drift comparison used in Theorem~\ref{thm:regret_mp_stability} proves strong stability for every $\bm\lambda\in\mathcal D_{\mathrm{prov}}^{\mathrm P}$.

\begin{proof}[Proof of Corollary~\ref{cor:stability_region_comparison}]
If $2(\varepsilon_n^{\mathrm{cov}}+\tau_n)\leq\bar\rho_n^{\mathrm P}$ for every $n$, then $a_n^{\mathrm{Sto}}\leq\bar\rho_n^{\mathrm P}$ and therefore $\delta^{\mathrm{Sto}}\leq\delta^{\mathrm P}$. Any comparison policy with reserve $\epsilon_0>\delta^{\mathrm P}$ also has $\epsilon_0>\delta^{\mathrm{Sto}}$, which gives $\mathcal D_{\mathrm{prov}}^{\mathrm P}\subseteq\mathcal D_{\mathrm{prov}}^{\mathrm{Sto}}$.

If $\varepsilon_n+\varepsilon_n^{\mathrm{cov}}+2\tau_n\leq\bar\rho_n^{\mathrm P}$ for every $n$, then $a_n^{\mathrm{DR}}\leq\bar\rho_n^{\mathrm P}$ and hence $\delta^{\mathrm{DR}}\leq\delta^{\mathrm P}$. In addition, $\varepsilon_n\geq\varepsilon_n^{\mathrm{cov}}$ implies $a_n^{\mathrm{Sto}}\leq a_n^{\mathrm{DR}}$ and therefore $\delta^{\mathrm{Sto}}\leq\delta^{\mathrm{DR}}$. The stated inclusions follow from the common reserve definition.
\end{proof}

\subsubsection{Fixed positive queue weighting and link-length normalization}
\label{app:weighted_stability}

Let $\bm W=\operatorname{diag}(w_m:m\in\mathcal M)$ satisfy $0<w_{\min}\leq w_m\leq w_{\max}<\infty$. The weighted network pressure is
$$
H_W(\bm s,t)=\bm Q^\top(t)\bm W[\bm I-\bm R^\top(t)]\bm C(t)\bm s,
$$
which decomposes as $H_W(\bm s,t)=\sum_nH_{n,W}(\bm s_n,t)$. Define the weighted phase regret, conditional pressure and regret, and the moduli $\Lambda_{n,W}(t)$ and $\Upsilon_{n,W}(t)$ by replacing $H_n$ with $H_{n,W}$ in Section~\ref{sec:network_stability}. The aggregation coefficients $a_n^\psi$ in~\eqref{eq:aggregation_pressure_coefficient} are unchanged because the same residual support and probability sets are used.

For the residual reconstructions considered in this paper, the same sensitivity argument yields finite constants $\beta_{n,W}\geq0$ and $h_{n,W}>0$ such that
$\Upsilon_{n,W}(t)\leq\beta_{n,W}+h_{n,W}\|\widehat{\bm Q}_{\mathcal L_n}(t)\|_1$.
Define
$$
\gamma_{n,W}=2(\beta_{n,W}+h_{n,W}b_{Q,n}),
\qquad
\ell_{n,W}=2h_{n,W}\eta_{Q,n},
$$
$$
B_{\xi,W}^\psi=\sum_na_n^\psi\gamma_{n,W},
\qquad
\delta_W^\psi=\max_{m\in\mathcal M}\sum_{n:m\in\mathcal L_n}a_n^\psi\ell_{n,W}.
$$
Then
\begin{equation}
    \sum_na_n^\psi\Lambda_{n,W}(t)
    \leq
    B_{\xi,W}^\psi
    +
    \delta_W^\psi\|\bm Q(t)\|_1.
    \label{eq:weighted_global_pressure_loss}
\end{equation}

\begin{lemma}[Stability with fixed positive queue weighting]
\label{lem:weighted_regret_stability}
Suppose Assumptions~\ref{ass:bounded_primitives}, \ref{ass:coverage}, and~\ref{ass:affine_reconstruction} hold, with the probability-coverage requirement omitted for Set Regret. If a comparison policy satisfies~\eqref{eq:conditional_service_reserve} with reserve $\epsilon_0$ and
\begin{equation}
    w_{\min}\epsilon_0>\delta_W^\psi,
    \label{eq:weighted_stability_condition}
\end{equation}
then weighted Regret MP strongly stabilizes the network. For a finite constant $B_W$ independent of the queue state,
\begin{equation}
    \limsup_{T\rightarrow\infty}\frac1T\sum_{t=0}^{T-1}\mathbb E\|\bm Q(t)\|_1
    \leq
    \frac{B_W+B_{\xi,W}^\psi}{w_{\min}\epsilon_0-\delta_W^\psi}.
    \label{eq:weighted_average_queue_bound}
\end{equation}
\end{lemma}

\begin{proof}[Proof of Lemma~\ref{lem:weighted_regret_stability}]
Use $V_W(\bm Q)=\frac12\bm Q^\top\bm W\bm Q$. The same expansion as in the unweighted case gives a finite constant $B_W$ such that
\begin{equation}
\mathbb E[\Delta V_W(t)\mid\mathcal H_t]
\leq
B_W+\bm Q^\top(t)\bm W\bm\lambda-H_W(\bm s(t),t).
\label{eq:weighted_drift_bound}
\end{equation}
The weighted conditional-pressure bound and~\eqref{eq:weighted_global_pressure_loss} give
$$
\mathbb E[H_W(\bm s^\psi(t),t)\mid\mathcal G_t]
\geq
\mathbb E[H_W(\widetilde{\bm s}(t),t)\mid\mathcal G_t]
-B_{\xi,W}^\psi
-\delta_W^\psi\mathbb E[\|\bm Q(t)\|_1\mid\mathcal G_t].
$$
From~\eqref{eq:conditional_service_reserve} and $\bm W\succeq w_{\min}\bm I$,
$$
\mathbb E\!\left[
H_W(\widetilde{\bm s}(t),t)-\bm Q^\top(t)\bm W\bm\lambda
\mid\mathcal G_t^Q
\right]
\geq
w_{\min}\epsilon_0\|\bm Q(t)\|_1.
$$
Taking the conditional expectation given $\mathcal G_t$, combining the preceding inequalities with~\eqref{eq:weighted_drift_bound}, and then taking expectations yields
$$
\mathbb E[\Delta V_W(t)]
\leq
B_W+B_{\xi,W}^\psi-(w_{\min}\epsilon_0-\delta_W^\psi)\mathbb E\|\bm Q(t)\|_1.
$$
Summation over time proves~\eqref{eq:weighted_average_queue_bound} under~\eqref{eq:weighted_stability_condition}.
\end{proof}

For link-length normalization, let movement $m=(i,j)$ have upstream-link length $L_i$ and choose $w_m=1/L_i$. If $0<L_{\min}\leq L_i\leq L_{\max}<\infty$, then $w_{\min}=1/L_{\max}$ and $w_{\max}=1/L_{\min}$. The weighted queue component is $X_{ij}=Q_{ij}/L_i$, so this result applies directly to the link-length-normalized pressure used in the microscopic experiments.

\section{Additional Proof for the Numerical Analysis}
\label{app:numerical_proof}

\subsection{Phase-decision invariance}
\label{app:decision_invariance_proof}

\begin{proof}[Proof of Proposition~\ref{prop:decision_invariance}]
For any competing phase $\bm u_n\neq\bm s_n^0$,
$$
\widetilde F_n(\bm u_n)-\widetilde F_n(\bm s_n^0)
\geq
F_n(\bm u_n)-F_n(\bm s_n^0)-2\vartheta_n
\geq
\Gamma_n-2\vartheta_n
>0.
$$
Thus every competing phase retains a strictly larger perturbed score than $\bm s_n^0$, so $\bm s_n^0$ remains the unique minimizer.
\end{proof}

\section{Detector-Based State Estimation and Uncertainty Calibration}
\label{app:microscopic_estimator}

This appendix gives the implementation details summarized in Section~\ref{sec:simulation}. The estimator uses aggregate physical detector counts only; vehicle identities are used by the simulator truth recorder for calibration and evaluation, not by the online estimator.

\subsection{Logical-link conservation and movement composition}

Let $\Delta_{\mathrm{det}}=15$~s be the detector update interval. For logical link $i$, let $\widehat Q_i(t)$ be its estimated vehicle total at the beginning of an update, $N_i^{\mathrm{in}}(t)$ the sum of observed entry-detector counts during the interval, and $N_{ij}^{\mathrm{out}}(t)$ the observed discharge assigned to movement $(i,j)$. The link total is updated by
\begin{equation}
    \widehat Q_i(t+\Delta_{\mathrm{det}})
    =
    \left[
        \widehat Q_i(t)
        +
        N_i^{\mathrm{in}}(t)
        -
        \sum_{j\in\mathcal J_i}
        N_{ij}^{\mathrm{out}}(t)
    \right]^+ ,
    \label{eq:app_link_conservation}
\end{equation}
where $\mathcal J_i$ is the set of downstream movement destinations from link $i$ and $[z]^+=\max\{z,0\}$. Thus, movement composition can redistribute the estimated link mass but cannot create or remove it.

Let $W_{\mathrm{mov}}=300$~s be the rolling movement-composition window and define the rolling movement count
\(N_{ij}^{\mathrm{roll}}(t) = \sum_{\tau\in(t-W_{\mathrm{mov}},t]} N_{ij}^{\mathrm{out}}(\tau).\)
When $\sum_{k\in\mathcal J_i}N_{ik}^{\mathrm{roll}}(t)>0$, the movement share is
\(\widehat\phi_{ij}(t) = \frac{N_{ij}^{\mathrm{roll}}(t)} {\sum_{k\in\mathcal J_i}N_{ik}^{\mathrm{roll}}(t)}.\)
If no movement count is available in the rolling window, the previous valid share is retained; the initial fallback is a topology-based nominal share. The movement state is then $\widehat Q_{ij}(t)=\widehat Q_i(t)\widehat\phi_{ij}(t)$, with $\sum_{j\in\mathcal J_i}\widehat Q_{ij}(t)=\widehat Q_i(t)$.

For an upstream movement $a$, let $\mathcal O_a$ be the set of its feasible downstream logical links. The routing construction used in the microscopic implementation distributes mass across these feasible downstream links and then uses the movement shares on the receiving link. If downstream movement $b=(j,k)$ lies on receiving link $j$, write $\widehat\phi_b(t):=\widehat\phi_{jk}(t)$. In compact form,
$$
\widehat r_{ab}(t)
    =
    \frac{1}{|\mathcal O_a|}
    \widehat\phi_b(t),
    \qquad
    b\ \text{belongs to a link in }\mathcal O_a,
$$
with infeasible transitions assigned zero and each routing row projected onto the feasible row-substochastic set when numerical reconciliation is needed.

\subsection{Length-normalized uncertainty state}

For movement $(i,j)$, define
\(X_{ij}(t)=\frac{Q_{ij}(t)}{L_i}, \qquad \widehat X_{ij}(t)=\frac{\widehat Q_{ij}(t)}{L_i}.\)
The residual stored for calibration is additive in the exact state used by the microscopic pressure calculation: $\epsilon_{ij}^{X}=\widehat X_{ij}-X_{ij}$ and $\epsilon_{ab}^{R}=\widehat r_{ab}-r_{ab}$. We collect these normalized-state and routing residuals in $\bm\epsilon$ and let $\bm E^R=[\epsilon_{ab}^R]$ denote the matrix form of the routing component.
For a retained residual scenario $\bar{\bm\epsilon}^{\,k}$, with routing-residual matrix $\bar{\bm E}^{R,k}$, the plausible latent state is reconstructed by
\begin{equation}
    X_{ij}^{k}
    =
    \left[
        \widehat X_{ij}
        -
        \bar\epsilon_{ij}^{X,k}
    \right]^+,
    \qquad
    \bm R^{k}
    =
    \Pi_{\mathcal R}
    \left[
        \widehat{\bm R}
        -
        \bar{\bm E}^{R,k}
    \right],
    \label{eq:app_clean_scenario_reconstruction}
\end{equation}
where $\Pi_{\mathcal R}$ is the Euclidean projection onto the feasible routing set. The reconstructed movement count passed to the common pressure evaluator is $Q_{ij}^{k}=L_iX_{ij}^{k}$. With fixed finite link lengths, the bounded additive residuals and nonnegative clipping in~\eqref{eq:app_clean_scenario_reconstruction} satisfy the affine reconstruction form in Assumption~\ref{ass:affine_reconstruction}.

Calibration uses three independent simulation runs over the interval $300$--$7200$~s, while formal baseline evaluation uses three disjoint runs. Residual components are scaled by a component-wise median absolute deviation factor before scenario reduction. Each local model retains $K_n=8$ representative joint residual vectors using weighted medoid reduction. The representative probabilities are the empirical assignment frequencies, and the operational DR radius is set to the 95th percentile of bootstrap Wasserstein distances generated from the retained scenario labels. This bootstrap-calibrated radius is the operational ambiguity radius used in the microscopic experiments; the coverage quantity $\varepsilon_n^{\mathrm{cov}}$ in Assumption~\ref{ass:coverage} is an analytical condition used in the stability result.

\subsection{Detector undercount}

Let $N_d(t)$ denote the unbiased interval count from detector $d$. For a detector assigned a fixed undercount fraction $b_d^{\mathrm{uc}}$, the observed count is $\widetilde N_d(t)=(1-b_d^{\mathrm{uc}})N_d(t)$.
The same biased physical-detector count is used wherever that detector enters the estimator. In the undercount experiment, $b_d^{\mathrm{uc}}=b_{\max}u_d$ with $u_d\in[0,1]$ drawn once for each selected detector and held fixed over the complete simulation. The detector sets are selected from the frozen topology and are independent of the traffic calibration data.

\subsection{Detector failure}

When a physical detector fails, its current-run count is not supplied to the estimator. Healthy counts remain fixed observations. Let $\mathcal D_f$ be a connected component of failed detectors entering the same local conservation system, let $\mathcal I_f$ be the set of logical links affected by that component, and let $\bm z$ collect the unknown interval flows of the failed detectors. For $i\in\mathcal I_f$, let $N_i^{\mathrm{in}}(t;\bm z)$ and $N_i^{\mathrm{out}}(t;\bm z)$ denote the reconstructed interval inflow and outflow after inserting the candidate failed-detector flows. The reconstruction uses a bounded $L_1$ conservation program of the form
$$
\begin{array}{ll}
\displaystyle \min_{\bm z,\bm v^+,\bm v^-,\bm e^+,\bm e^-} & \displaystyle \sum_{d\in\mathcal D_f}w_d(v_d^++v_d^-)+\eta_{\mathrm f}\sum_{i\in\mathcal I_f}(e_i^-+e_i^+)\\[2pt]
\text{s.t.} & z_d-\bar z_d=v_d^+-v_d^-,\quad d\in\mathcal D_f,\\
& 0\le z_d\le\bar z_d^{\max},\quad d\in\mathcal D_f,\\
& v_d^+,v_d^-\ge0,\quad d\in\mathcal D_f,\qquad e_i^+,e_i^-\ge0,\quad i\in\mathcal I_f,\\
& 0\le\widehat Q_i(t)+N_i^{\mathrm{in}}(t;\bm z)-N_i^{\mathrm{out}}(t;\bm z)+e_i^--e_i^+\le Q_i^{\max},\quad i\in\mathcal I_f.
\end{array}
$$
where $v_d^+$ and $v_d^-$ are the nonnegative positive and negative deviations of $z_d$ from its historical proxy $\bar z_d$, $w_d>0$ weights that deviation, $\eta_{\mathrm f}>0$ weights conservation slack, $\bar z_d^{\max}$ is a calibrated upper bound, and $Q_i^{\max}$ is the physical logical-link storage bound. The nonnegative slacks $e_i^-$ and $e_i^+$ preserve feasibility under temporary inconsistency. Reconstructed flows are blended with historical movement-share and service information before being passed to the same conservation estimator in~\eqref{eq:app_link_conservation}. Network-entry detectors are protected in the formal failure experiment so that external inflows remain directly observed.

\end{document}